\documentclass[lettersize,journal]{IEEEtran}

\usepackage{cite}
\usepackage{amsmath,amssymb,amsfonts}
\usepackage{algorithmic}
\usepackage{amsthm}
\usepackage{algorithm}
\usepackage{array}
\usepackage[caption=false,font=normalsize,labelfont=sf,textfont=sf]{subfig}
\usepackage{textcomp}
\usepackage{stfloats}
\usepackage{url}
\usepackage{color,soul}
\usepackage{verbatim}
\usepackage{graphicx}
\usepackage{hyperref}
\usepackage{threeparttable}
\usepackage{tabularx,booktabs}
\newtheorem{theorem}{Theorem}

\newtheorem{lem}{Lemma}[theorem]
\newtheorem{corollary}{Corollary}[lem]

\newcommand\tc[1]{\textcolor{blue}{#1}}
\newcolumntype{C}{>{\centering\arraybackslash}X} 
\usepackage{lipsum} 

\def\BibTeX{{\rm B\kern-.05em{\sc i\kern-.025em b}\kern-.08em
		T\kern-.1667em\lower.7ex\hbox{E}\kern-.125emX}}

\begin{document}
	
	\title{Design and Implementation of Energy-Efficient Wireless Tire Sensing System with Delay Analysis for Intelligent Vehicles}

	\author{Shashank Mishra, and Jia-Ming Liang, \IEEEmembership{Member, IEEE}
		\thanks{This research is co-sponsored by NSTC 109-2221-E-024-012-MY3 and NSTC 109-2221-E-024-013-MY3. This paper was presented in part at the 25th International Computer Symposium, ICS 2022 \cite{mishra2023_1} in Taoyuan, Taiwan and was published in the corresponding proceedings. 
			\textit{(Corresponding authors: Jia-Ming Liang)} }
		\thanks{S. Mishra and J.-M. Liang are with the Department of Electrical Engineering, National University of Tainan, Tainan, 70005 Taiwan (e-mail: d10982003@stumail.nutn.edu.tw; jmliang@mail.nutn.edu.tw).}}
	
	\markboth{}%
	{S. Mishra \MakeLowercase{\textit{et al.}}: Design and Implementation of Energy-Efficient Wireless Tire Sensing System with Delay Analysis for Intelligent Vehicles}
	
	
	\maketitle
	
	\begin{abstract}
		The growing prevalence of Internet of Things (IoT) technologies has led to a rise in the popularity of intelligent vehicles that incorporate a range of sensors to monitor various aspects, such as driving speed, fuel usage, distance proximity and tire anomalies. Nowadays, real-time tire sensing systems play important roles for intelligent vehicles in increasing mileage, reducing fuel consumption, improving driving safety, and reducing the potential for traffic accidents. However, the current tire sensing system drains a significant vehicle' energy and lacks effective collection of sensing data, which may not guarantee the immediacy of driving safety. Thus, this paper designs an energy-efficient wireless tire sensing system (WTSS), which leverages energy-saving techniques to significantly reduce power consumption while ensuring data retrieval delays during real-time monitoring. Additionally, we mathematically analyze the worst-case transmission delay of the system to ensure the immediacy based on the collision probabilities of sensor transmissions. This system has been implemented and verified by the simulation and field trial experiments. These results show that the proposed scheme provides enhanced performance in energy efficiency 
		and accurately identifies the worst transmission delay. 
		
	\end{abstract}
	
	
	\section{Introduction}
	\IEEEPARstart{T}{he} demand for Internet of Things (IoT)-enabled real-time monitoring devices is on the rise in the field of automotive industry, where safety is of the utmost importance \cite{c1,c2}.
	Currently, the regulations of most countries such as the United States, the European Union, and Korea \cite{USA, EU, KR} have mandated that vehicles equipped with tire pressure monitoring systems be roadworthy, especially for passenger vehicles and/or larger vehicles such as trucks, flatbeds, and trailers, since accidents involving these vehicles often result in significant and severe consequences \cite{RRRR1, EE}.
	Numerous wired and wireless technologies employing radio frequency (RF) receivers have been introduced to address the issue of real-time tire status monitoring \cite{c6,c8}.
	The real-time data reception is crucial for vehicle safety. Rapid tire pressure drops, such as those from hitting a pothole or due to pressure leakage, require an immediate response to warn the driver and potentially activate safety features \cite{RR1}. However, the effectiveness of RF receivers is limited by the coverage range of their antennas and the need for close proximity to mitigate sensor transmission delay \cite{c9}. Additionally, traditional systems continuously detect sensing data, leading to significant energy consumption in vehicles, especially for larger vehicles such as trucks, flatbeds, and trailers. These larger vehicles are often detached from their tractors and parked separately in parking lots, making it difficult to sustain power supply from the tractor’s battery \cite{I_1, I_2, I_3}. Therefore, it is necessary to identify scenarios of sensor transmission delay while considering energy efficiency to ensure immediacy for driving safety.
	
	To address the above issue, we propose an energy-efficient IoT-enabled wireless tire sensing system (WTSS) that incorporates wireless sensors and multiple RF transmission modules to enable real-time monitoring. We develop and demonstrate this system based on wireless pressure and temperature sensors, which are widely used in intelligent vehicles \cite{c10}. These sensors are low-end and typically comprise a radio-frequency (RF) transmitter module, sensor cells, and a unique sensor ID, which make it possible for the receiver end to decode the pressure and temperature in real time \cite{c12}.
	In addition, an energy-saving mechanism is designed at the system receiver to conserve energy. This mechanism enforces the receiver to perform wake-up and sleep operations to maintain system reliability and immediacy. Specifically, an analytic scheme is developed to identify the wake-up and sleep schedule and also analyze the sensing delays to achieve real-time identification and avoid potential consequences. The extended simulations and experiments of the field trial will show that the proposed system can save
	energy 
	while ensuring the worst transmission delay.
	
	The major contributions of this paper are fourfold.
	\begin{itemize}
		
		\item{First, to the best of our knowledge, this is the first paper to address the issue of energy and sensing delay in the intelligent tire sensing system. We point out the necessity of energy conservation issues for intelligent vehicle systems and propose an innovative, comprehensive approach encompassing an energy-saving mechanism with sleep behavior analysis, sensor transmission delay evaluation, and timely data scheduling.}
		
		\item{Second, we propose a high-accuracy \emph{analytic scheme} and prove its properties by 2 \emph{Theorems}, 4 \emph{Lemmas}, and 2 \emph{Corollaries} to precisely determine the worst transmission delay accounting for sensor collision probability, sensor reception rate, and power saving ratio for multiple sensors transmitting simultaneously. This novel scheme effectively enhances system reliability for real-time sensor data transmission while ensuring energy efficiency.}
		
		\item{Third, we have designed and implemented this prototype system and conducted extensive field trial experiments in a large parking lot to validate the system's effectiveness and feasibility. Comprehensive real-time experiments were also conducted with a variety of scenarios. Their results demonstrate the prototype's feasibility and reliability and also reveal that the proposed scheme can accurately predict and ensure sensor transmission delays while significantly saving energy. This approach is equally essential for other real-time systems where timely data is crucial for safety and optimal performance.}
		
		\item{Fourth, the proposed architecture can also be applied to other wireless sensor-based information gathering in intelligent vehicles, such as PIR sensors, wheel inertial sensors, ultrasonic sensors, and/or radar sensors \cite{MC4_1, MC4_2, MC4_3, MC4_4}. These can all benefit from the proposed methodologies.}
	\end{itemize}
	
	The rest of this paper is structured as follows. Section II is related work. Section III introduces the system overview, including system architecture and problem identification. Section IV describes the proposed scheme and proofs. Section V shows the performance evaluation. Section VI presents the system demonstration, and Section VII draws the final conclusions and outlines future work.
	
	\section{Related work}
	Significant efforts have been made in previous years to develop tire status monitoring systems for improving safety and preventing accidents \cite{c13,c14}.
	To address sudden changes in tire status, different technologies and concepts have been proposed \cite{c15,c16}. Whereas several studies have explored the potential of wireless sensor networks and Internet of Things (IoT) technologies to enhance the performance of different IoT-based monitoring systems \cite{c21,c22}.
	Specifically, reference \cite{c17} uses a pressure sensor and voice warning system on a helmet to measure tire air pressure and transmit data via Bluetooth. However, it overlooks tire temperature and potential transmission delays. Reference  \cite{c19} designs an ID calibration device for a tire pressure monitoring system using Bluetooth and CAN bus data interaction, with a time-based task scheduler algorithm and a dual producer-consumer buffer. Reference  \cite{c20} proposes a basic monitoring system for tire mileage and wear estimation using G-sensors and Hall sensors, relying solely on Bluetooth for data collection and user interaction, neglecting the importance of other wireless technologies for vehicle safety.
	The work \cite{t1} presents a smart tire status monitoring system that detects vehicle mass but does not analyze sensor data to identify tire conditions. The study \cite{t2} designs an air pressure detection system using an IoT platform but fails to propose practical approaches to solve significant issues or enhance system performance. Reference \cite{t3} presents a Vehicle-to-Cloud interface utilizing 3G/4G connectivity and LoRaWAN. Reference \cite{t4} proposes an IoT-based fleet monitoring system utilizing cloud computing. However, both studies  \cite{t3,t4} neglect issues related to sensor transmission delay and system reliability.
	The study \cite{c11} proposes a hybrid tire status monitoring system that estimates the pressure of all four wheels using a single pressure sensor. The work \cite{c18} develops an IoT-based TPMS using separate sensors for measuring air pressure and temperature in real time. Reference \cite{t5} presents an indirect tire pressure monitoring system that uses existing sensors for vehicle dynamics control to detect tire pressure. The proposed system uses a modified transceiver structure with a temperature-compensated film bulk acoustic resonator (FBAR) and lowers the energy consumption of the wireless receiver system. However, these studies \cite{c11,c18,t5} do not consider real-time monitoring systems that utilize multiple wireless communication modes and analyze system reception delays.
	
	Furthermore, several studies \cite{R1, R2, R3, R5, R6} address the energy consumption issue for vehicle monitoring systems, particularly focusing on the sensor transmitter and receiver. Reference \cite{R1} proposes a method for identifying the vehicle's running state, with an emphasis on developing small-size and low-energy consumption systems. Reference \cite{R2} introduces a vehicle health monitoring system for accessing real-time information. Study \cite{R3} examines the energy consumption problem in connected vehicular applications and communication systems. Study \cite{R5} investigates the impact of sensor signal frequency on the power use of sensor transmitters. Study \cite{R6} examines tire slip angle estimation based on tire pressure sensors affecting transmitter power consumption.
	However, these existing studies \cite{R1, R2, R3, R5, R6} focused entirely on addressing the energy consumption issue while neglecting to consider sensor transmission delay.
	
	Hence, there is a need to design and implement an efficient IoT-enabled system for real-time tire status monitoring, energy-saving consideration, analysis of sensor transmission delay, and ensuring system reliability in intelligent vehicle systems. 
	\begin{figure} [!h]
		\centering
		\includegraphics[width=\linewidth]{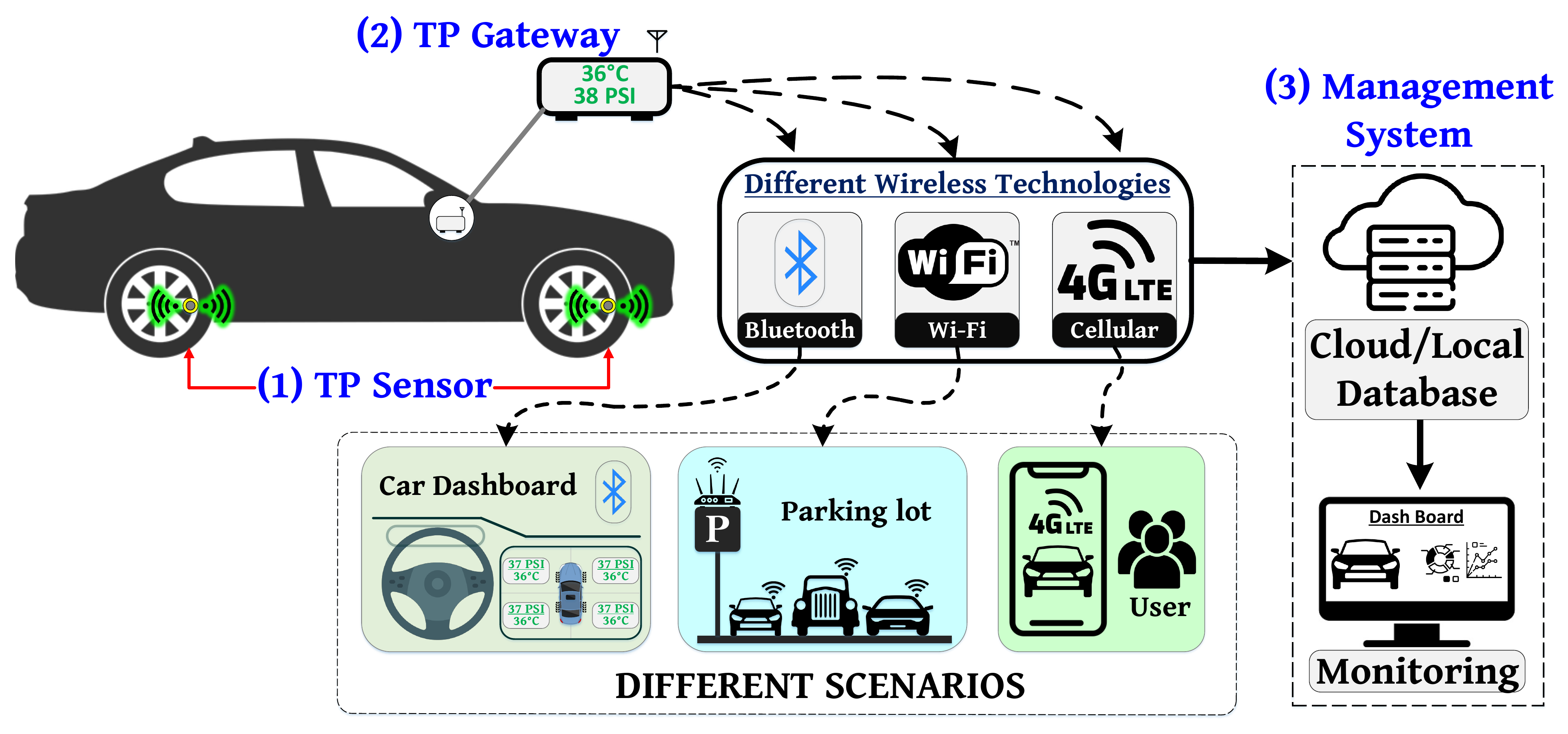}
		\caption{IoT-Enabled Wireless Tire Sensing System (WTSS)}
		\label{WTPSS}
	\end{figure}
	\section{System Overview}
	In this paper, we propose and develop a novel IoT-enabled wireless tire sensing system. The overview of the system is shown in Fig. \ref{WTPSS}. This system comprises wireless TP (temperature and pressure)\footnote{Tire temperature and pressure frequently influence each other \cite{c2, FN1}. Monitoring tire condition can lead to various benefits, including accident reduction, fuel consumption minimization, driving convenience improvement, braking distance reduction, and tire lifespan extension \cite{FN2, FN3}. It is worth noting that certain tire monitors also incorporate inertial components \cite{FN4, FN5}. This paper specifically concentrates on tire sensors measuring pressure and temperature exclusively.} sensors strategically positioned on each tire of the vehicle and a TP Gateway installed within the vehicle.
	
	The TP gateway simultaneously sends data to local and cloud servers via wireless communication systems, including Wi-Fi, Bluetooth, and cellular networks, depending on the required scenarios. 
	Specifically, this system consists of three main components: (1) Wireless TP Sensor, (2) TP Gateway, and (3) Management System, as shown in Fig.~\ref{architecture}.
	These components collaborate to facilitate seamless data reception, analysis, conversion, and real-time monitoring, ensuring efficient data processing and comprehensive management.
	\begin{figure} [t!]
		\centering
		\includegraphics[width=\linewidth]{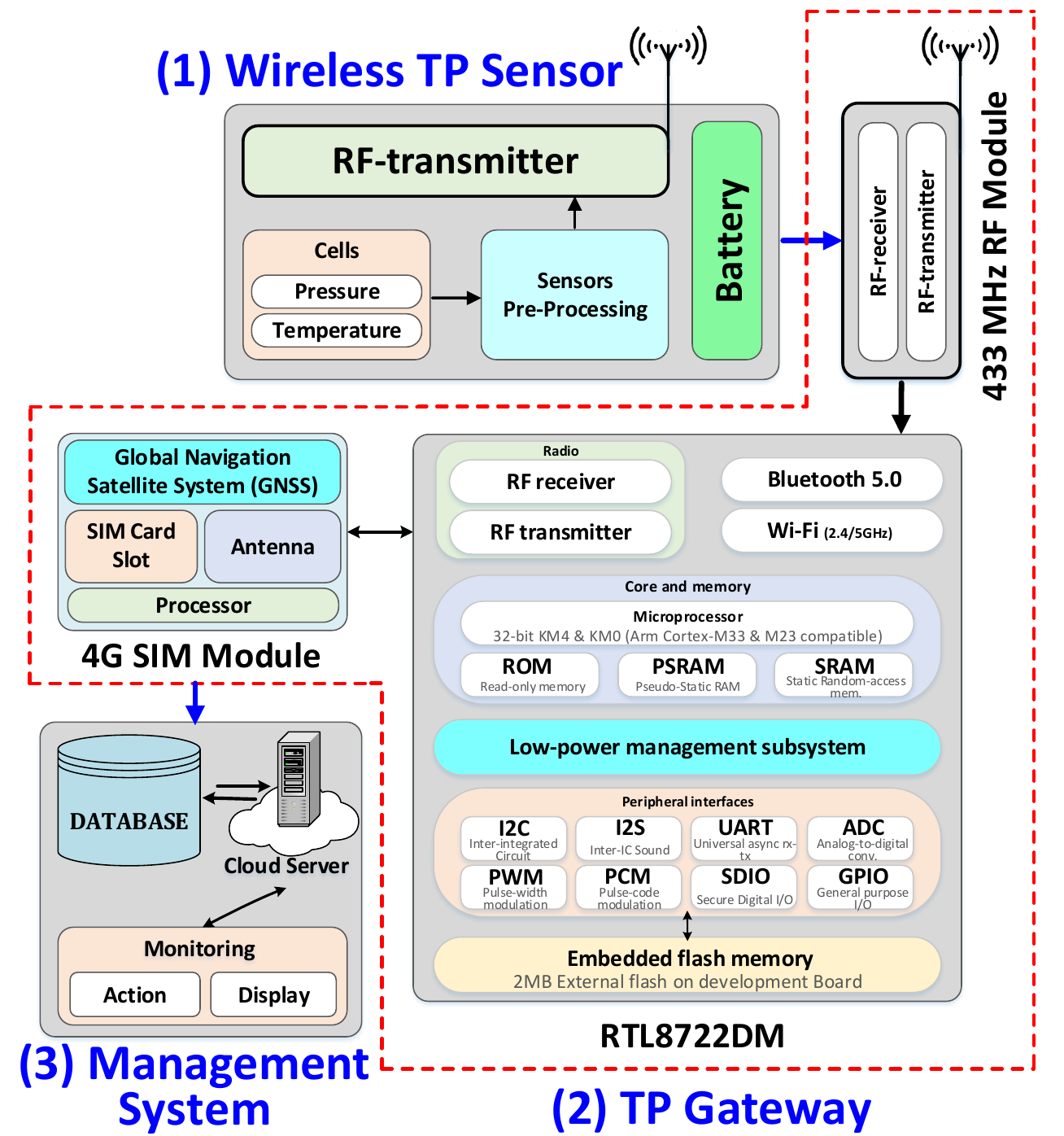}
		\caption{System Architecture of IoT-enabled WTSS}
		\label{architecture}
	\end{figure}
	
	1) Wireless TP Sensor \footnote{Currently, TP Sensor can continue to operate for several years, lasting until tire replacement, due to their low-end and uncontrollable design \cite{R4, R5, R6}. Disrupting the sensor's energy cycle or external interference could lead to unreliable data transmission or malfunction \cite{RRR1, RRR2}.} is battery powered and incorporates modules to collect accurate tire pressure and temperature data to wirelessly transmit the data through a 433 MHz antenna. The TP sensor operates on a fixed transmission duty-cycle (with the length of $C_L$ (in slots)) \cite{r_sensor} and periodically transmits data containing sensor ID, pressure, temperature, alert status, battery level, and sensor status, following a predefined format, as shown in Fig.~ \ref{ds}.
	\begin{figure} [!h]
		\centering
		\includegraphics[width=\linewidth]{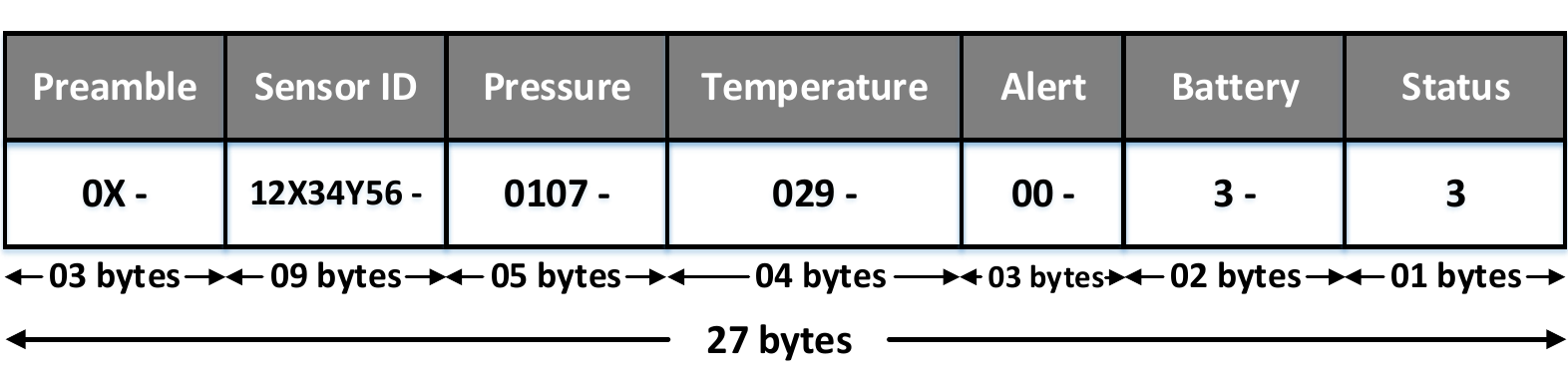}
		\caption{Example data format of a TP sensor \cite{r_sensor}}
		\label{ds}
	\end{figure}
	
	2) TP Gateway comprises three main modules, namely (a) 433 MHz RF Receiver Antenna, (b) RTL8722DM IoT Development Board, and (c) 4G SIM Module. Specifically, (a) 433 MHz RF receiver Antenna receives the encoded sensor data and transmits the signals to the RTL8722DM IoT development board for decoding into a suitable format. (b) RTL8722DM IoT Development Board is a 32-bit dual micro-controller unit and features dual-band Wi-Fi and BLE5 mesh, making it ideal for IoT applications. (c) 4G SIM Module provides extensive network coverage for 4G/LTE, 3G, and 2G bands, accurate positioning, and navigation with its built-in GNSS receiver.
	
	3) Management system represents the connection between a physical device and the platform service, encompassing a database for local and web servers that grant access to a Graphical User Interface (GUI) and mobile application. 
	
	The data transmission from the sensor to the TP Gateway is illustrated in Fig. \ref{problem}. In order to save system energy, the TP gateway is designed and enforced to perform wake-up (in $W$ slots) and sleep (in $S$ slots) to monitor the sensing data. Thus, it constitutes a sleep behavior with a cycle length of $W+S$ slots. Periodically, each sensor ($s_{i}, i = 1, 2, ..., N$) transmits data every $C_{L}$ slots to the system through the 433 MHz RF antenna. The TP gateway monitors the transmitted data during the designated wakeup slot ($W$). 
	If the data transmission collides with another transmission (denoted by $C_t$, i.e., the cross slot in the figure, where $t$ is the number of collisions and $t \in \mathbb{N}$) or occurs during the system sleeping slot, it may not be received by the system, leading to an expected delay in the receipt of the transmitted data.
	
	\begin{figure} [!h]
		\centering
		\includegraphics[width=\linewidth]{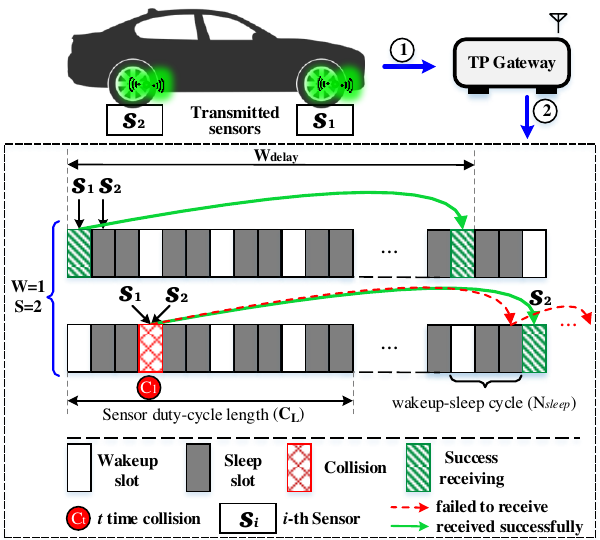}
		\caption{illustration of data transmission to the TP Gateway}
		\label{problem}
	\end{figure}
	
	In order to ensure timely data reception in the system and mitigate potential consequences of delays, it is crucial to address system reliability and data transmission delays by asking the following questions:
	
	\begin{enumerate}
		\item Does the system exhibit a finite maximum waiting time?
		\item How to quantify the total expected worst transmission delay ${E[W_{\text{delay}}^{\text{total}}]}$ of this system considering the collision probability?
		\item Can it effectively determine the optimal transmission parameters, such as the minimum sensor duty-cycles $C^\text{min}_L$ and the minimum number of wakeup-sleep cycles $N^{{min}}_{sleep}$ by considering additional factors to improve the system's reliability?
	\end{enumerate}
	
	Table \ref{table2} presents the notations used in this paper.
	
	\begin{table}[!h]
		\caption{Notations and Description}
		\label{table2}
		\begin{tabularx}{\columnwidth}{@{} ll *{2}{C} c @{}}
			\toprule
			\bf{Notations } & \bf{Description} \\ 
			\midrule
			
			$C_{L}$ & Duty-cycle length of a sensor transmission \\
			$C_{L}^{i}$ & The $i^{th}$ duty-cycle length of sensors \\
			$C_{L}^{min}$ & Minimum number of sensor duty-cycles to successfully \\
			$\quad$ & receive sensor transmission \\
			$E[W_{delay}^{total}]$ & Total expected worst transmission delay \\
			$N$ & Total number of sensors \\
			$N_{sleep}$ & Number of wakup-sleep cycles \\
			$N_{sleep}^{min}$ & Minimum number of wakup-sleep cycles to successfully \\
			$\quad$ & receive sensor transmission \\
			$Pr(c)$ & Probability of collision \\
			$Pr(s)$ & Probability of success \\
			$S$ & Sleeping length (in slots) \\
			$s_i$ & The $i^{th}$ sensor in the system \\
			$T_{n}$ & The sensor arrives at the $n^{th}$ time slot \\ 
			$T^{max}$ &  Worst sensor transmission \\ 
			$W$ & Wakeup length (in slots) \\
			$W_{delay}$ & Worst transmission delay \\
			$W+S$ & Wakeup-sleep cycle length \\
			\bottomrule
		\end{tabularx}
	\end{table}
	
	\section{The Proposed Analytic Scheme}
	In this section, we present an \emph{analytic scheme} for determining the worst transmission delay ($W_{delay}$) according to the worst sensor transmission ($T^{max}$) and minimum number of wakeup-sleep cycles ($N_{sleep}^{min}$). In order to identify $T^{max}$ and $N_{sleep}^{min}$, we considered various combinations of sleeping ($S$) slots, where wakeup $(W)$ is fixed as one slot for scheduling data in order to optimize energy consumption of the system. The scheme includes the number of transmitting sensors ($s_{i},i=1, 2, ..., N$), the number of wakeup-sleep cycles ($N_{sleep}$), the sensor duty-cycle length ($C_{L}$), and the worst sensor transmission $T^{max}$.
	
	\subsection{Analytic Scheme}
	\label{det_sch}
	In this scheme, we consider two cases: $Case$ $1$ is when $GCD(C_{L}, [W+S])\neq 1$ whereas $Case$ $2$ is when $GCD(C_{L}, [W+S])=1$. Additionally, we also consider sensor transmission collision scenarios to quantify the total expected worst transmission delay $E[W_{delay}^{total}]$. 
	
	First, we need to find the sensor incurring the worst transmission delay. Without loss of generality, we assume it first arrives at the $n^{th}$ slot, denoted by $T_n$, where $1 \leq n \leq C_L$. The initial step involves identifying whether the sensor $T_{n}$ can be received under the wakeup-sleep operations or not. We consider the wakeup-sleep cycle starting with a wakeup slot. Then, finding the worst sensor transmission $T^{max}$ can be modeled as the scenario where a sensor arrives at the $n$-th slot but fails to receive until several duty-cycles passed, i.e., $n+C^{\text{min}}_L\cdot C_L$, and is finally received at a $W$ slot over several wakeup-sleep cycles, i.e., $W+N^{\text{min}}_{sleep} \cdot[W+S]$. This is expressed as follows.
	\begin{equation}
		\begin{aligned}
			T^{\text{max}}=n^*=\mathop{\arg \max}\limits_{n=1,2,...,C_{L}} \{ n+C^{\text{min}}_L\cdot C_L \mid n+C^{\text{min}}_L \cdot C_L \\= W+N^{\text{min}}_{sleep}\cdot[W+S]\}.
			\label{tmax}
		\end{aligned}
	\end{equation}
	Here, $C^{\text{min}}_L$ and $N^{min}_{sleep}$ are the minimum number of sensor duty-cycles and minimum number of wakeup-sleep cycles that a sensor transmission $T_n$ arrives to meet a wakeup slot, respectively. 
	
	In the following, we show that $T^{max}$ exists only when the sensor duty-cycle length $C_L$ and the wakeup-sleep cycle length $[W+S]$ are relative prime. This is expressed in Theorem \ref{th001}:
	\begin{theorem} The system experiences worst sensor transmission upon successful sensor transmission reception, if 
		\begin{equation}
			T^{\text{max}} =
			\left\{
			\begin{array}{ll}
				C_{L}-S & \mbox{if }GCD(C_{L}, [W+S]) = 1\\
				\textit{Not exist}  & \mbox{if }GCD(C_{L}, [W+S]) \neq 1 
			\end{array}
			\right..
			\label{th01}
		\end{equation}
		\label{th001}
	\end{theorem}
	We employ the following Lemmas $\ref{lemma1}\sim \ref{lemma4}$ to constitute Theorem \ref{th001}.
	\begin{lem}
		\label{lemma1}
		The system exists a feasible number of sensor duty-cycles and number of wakeup-sleep cycles to successfully receive all possible arrivals of sensor transmissions if $C_L$ and $[W+S]$ are relative prime, i.e., if $GCD(C_{L}, [W+S]) = 1$,  it exists minimal integers $C^{\text{min}}_L \in \mathbb{N}$ and $N^{\text{min}}_{sleep} \in \mathbb{N}$ that satisfies the equation $n+C^{\text{min}}_L \cdot C_L=W+ N^{\text{min}}_{sleep}\cdot [W+S]$ for all $n=1,2,....,C_{L}$.
	\end{lem}
	\begin{proof}[\bfseries{Proof}]
		We use a direct proof method as follows. First, we rewrite Eq. \eqref{tmax} as follows:
		\begin{equation}
			\begin{aligned}
				N^{\text{min}}_{sleep}\cdot[W+S]-C^{\text{min}}_L\cdot C_L= n-W.
				\label{p01}
			\end{aligned}
		\end{equation}
		Since $GCD(C_{L}, [W+S]) = 1$, it means that $C_L$ and $[W+S]$ can be represented as a linear combination resulting in $1$ according to Bézout’s theorem \cite{R7}, i.e.,
		\begin{equation*}
			\begin{aligned}
				GCD(C_{L}, [W+S]) = 1=\alpha\cdot[W+S]+\beta\cdot C_L,
				\label{p02}
			\end{aligned}
		\end{equation*}
		where $	\exists\alpha, 	\exists\beta \in \mathbb{Z}$. Thus, for each $n=1,2,....,C_{L}$, we can find an integer multiple of $\alpha$ and $\beta$ to satisfy Eq. \eqref{p01}, i.e., $N^{\text{min}}_{sleep}\cdot[W+S]-C^{\text{min}}_L\cdot C_L$
		\begin{equation}
			\begin{aligned}
				&= (n-W)\\
				&=(n-W)(\alpha\cdot[W+S]+\beta\cdot C_L) \\
				&=((n-W)\cdot \alpha)[W+S]+((n-W)\cdot \beta )C_L.
				\label{p04}
			\end{aligned}
		\end{equation}
		This means that $N^{\text{min}}_{sleep}$ and $C^{\text{min}}_L$ have feasible solutions to satisfy Eq. \eqref{tmax} when $N^{\text{min}}_{sleep}=(n-W)\cdot \alpha$ and $C^{\text{min}}_{L}=-(n-W)\cdot \beta$ for all $n=1,2,....,C_{L}$. Since $\alpha$, $\beta$ have multiple pairs, by well-ordering principle \cite{well}, it exists a pair incurring the minimal and positive solutions of $N^{\text{min}}_{sleep}\in \mathbb{N}$ and $C^{\text{min}}_L \in \mathbb{N}$ to satisfy. Under this condition, the system has a "finite" waiting time to receive all possible sensor transmissions.
	\end{proof}
	\begin{lem}
		\label{lemma2}
		When the system has a finite waiting time, the minimum number of sensor duty-cycles and minimum number of wakeup-sleep cycles required to successfully receive the worst sensor transmissions are $C^{\text{min}}_L=S$ and $N^{\text{min}}_{sleep}=C_L-1$ when $n=C_L-S$, under the condition that $C_L>S$ and $C_L$ and $[W+S]$ are relative prime.
	\end{lem}
	\begin{proof}[\bfseries{Proof}]
		First, we show that when $n= C_L-S$, it has $C_L^{min}=S$. We use a direct proof as follows. According to Eq. \eqref{tmax}, it can be written as
		\begin{equation}
			\begin{aligned}
				N^{\text{min}}_{sleep}=\frac{n+C^{\text{min}}_L\cdot {C_L}-W}{[W+S]} \in \mathbb{N}.
				\label{eqlem2}
			\end{aligned}
		\end{equation}
		By applying $n=C_L-S$, we have	
		\begin{equation}
			\begin{aligned}
				Eq. \eqref{eqlem2} &\Rightarrow \frac{(C_L-S)+C^{\text{min}}_L{C_L}-W}{[W+S]}\\ &=\frac{C_L+C^{\text{min}}_L{C_L}-{(W+S)}}{[W+S]}.
				\label{lem1.2}
			\end{aligned}
		\end{equation}
		Because we have to ensure the result of Eq. \eqref{lem1.2} belonging to $\mathbb{N}$ when $W=1$ is set up in this scheme, it means that the following equation should belong to $\mathbb{N}$ as wall, i.e.,
		\begin{equation}
			\begin{aligned}
				Eq. \eqref{lem1.2} \Rightarrow \frac{C_L(C^{\text{min}}_L+1)}{[1+S]} \in \mathbb{N}.
				\label{lem1.2.2}
			\end{aligned}
		\end{equation}
		Since $GCD(C_L,[W+S])=1$, it means that Eq. \eqref{lem1.2.2} has to satisfy $[W+S]\mid C^{\text{min}}_L+1$. As we know that the minimum number of sensor duty-cycles should be larger or equal to $0$, i.e., $C^{\text{min}}_L\geq0$. Thus, the smallest value of $C^{\text{min}}_L$  to satisfy Eq. \eqref{lem1.2.2} will be $C^{\text{min}}_L=S$ because $C_L$ and $[1+S]$ have no common factor and $C^{\text{min}}_L$ cannot satisfy if $C^{\text{min}}_L<S$. Based on this result, it yields
		\begin{equation*}
			\begin{aligned}
				N^{\text{min}}_{sleep}=\frac{(C_L-S)+S\cdot C_L-1}{[1+S]}&=\frac{C_L(1+S)-(1+S)}{[1+S]}\\
				&=C_L-1.
				\label{lem1.2.3}
			\end{aligned}
		\end{equation*}
	\end{proof}
	Before introducing \emph{Lemma} \ref{lemma3}, we show the following two corollaries first.
	\begin{corollary}
		\label{cor1}
		When the system has a finite waiting time, the minimum number of duty-cycle $C^{\text{min}}_L$ will not exceed $S$, i.e., $C^{\text{min}}_L\leq S$ if $GCD(C_L, [W+S])=1$.
	\end{corollary}
	\begin{proof}[\bfseries{Proof}]
		According to Eq. \eqref{eqlem2}, the numerator can be separated as two parts: $(n-1)$ and $(C^{\text{min}}_L \cdot C_L)$.
		\begin{equation*}
			\begin{aligned}
				\frac{(n-1)+(C^{\text{min}}_L \cdot C_L)}{[1+S]} \in \mathbb{N},
			\end{aligned}
		\end{equation*}
		it means that any remainder of $(C^{\text{min}}_L \cdot C_L)$ divided by $(1+S)$ should be deducted by the remainder of $(n-1)$ divided by $(1+S)$. Without loss of generality, we let $\theta \in [\theta,..,S]$ be the remainder of $(C^{\text{min}}_L \cdot C_L)$ divided by $(1+S)$. The above property can be written as
		\begin{equation}
			\begin{aligned}
				C^{\text{min}}_L \cdot C_L\mod{(1+S)}=\theta
				\label{col1.1}
			\end{aligned}
		\end{equation}
		and	
		\begin{equation}
			\begin{aligned}
				n-1\mod{(1+S)}=-\theta.
				\label{col1.2}
			\end{aligned}
		\end{equation}
		Also, we let $r$ be the remainder of $C_L$ divided by $(1+S)$, i.e., $C_L=q(1+S)+r$, where $q\in \mathbb{N}$ is a quotient. We can rewrite Eq. \eqref{col1.1} as
		\begin{align}
			\theta&=C^{\text{min}}_L \cdot C_L \mod (1+S)\notag \\
			&=C^{\text{min}}_L \cdot (r) \mod (1+S).\notag \\
			&\Rightarrow C^{\text{min}}_L \cdot C_L \equiv C^{\text{min}}_L \cdot (r) \equiv \theta \mod (1+S).
			\label{col1.3}
		\end{align}
		Based on Eq. \eqref{col1.3}, Eq. \eqref{col1.2} can be written as
		\begin{equation*}
			\begin{aligned}
				n-1\mod{(1+S)}&=-\theta \\
				&=-(C^{\text{min}}_L \cdot (r) \mod (1+S)) \\
				&=C^{\text{min}}_L \cdot (-r) \mod (1+S).
			\end{aligned}
		\end{equation*}
		\begin{equation}
			\begin{aligned}
				\Rightarrow n-1\equiv C^{\text{min}}_L \cdot (-r) \equiv -\theta \mod (1+S).
				\label{col1.4}
			\end{aligned}
		\end{equation}
		This can conclude two important properties. First, for those $n$ satisfying Eq. \eqref{col1.4}, they have the same $C^{\text{min}}_L$, i.e., 
		\begin{equation*}
			\begin{aligned}
				n-1\equiv C^{\text{min}}_L \cdot (-r) \equiv -\theta \mod (1+S)
			\end{aligned}
		\end{equation*}
		\begin{equation}
			\begin{aligned}
				\Rightarrow n=k(1+S)-\theta + 1, k \in \mathbb{N}^+.
				\label{col1.5}
			\end{aligned}
		\end{equation}
		Second, from Eq. \eqref{col1.4}, $C^{\text{min}}_L$ will not exceed $S$. We show it as follows.
		
		Assume there is a $C^{\text{min}'}_L >S$ and we can find a $C^{\text{min}}_L\leq S$ to represent as $C^{\text{min}'}_L =\Delta (1+S)+C^{\text{min}}_L$, where $\Delta \in \mathbb{N}$. From Eq. \eqref{col1.4}, we can derive it as follows.
		\begin{equation*}
			\begin{aligned}
				C^{\text{min}'}_L \cdot (-r) &\mod (1+S) \\
				&=(\Delta(1+S)+C^{\text{min}}_L)(-r) \mod (1+S) \\
				&=C^{\text{min}}_L(-r) \mod (1+S) \\
				&\Rightarrow C^{\text{min}'}_L(-r) \equiv C^{\text{min}}_L (-r) \mod (1+S).
				\label{col1.6}
			\end{aligned}
		\end{equation*}
		Thus, for those $C^{\text{min}'}_L >S$, the results will be congruent to that of $C^{\text{min}}_L\leq S$. So, we have the minimum number of duty-cycle $C^{\text{min}}_L\leq S$. 	
	\end{proof}
	\begin{corollary}
		\label{cor2}
		For all $C^{\text{min}}_L\in [0,..,S]$, their corresponding series of $n$ to satisfy Eq. \eqref{eqlem2} will constitute all possible transmission arrivals.
	\end{corollary}
	\begin{proof}[\bfseries{Proof}]
		First, we show that for any pair of $C^{\text{min}'}_L,C^{\text{min}}_L \in [0,..,S]$ and $C^{\text{min}'}_L \neq C^{\text{min}}_L$, their corresponding $\theta'$ and $\theta$, respectively, will not be the same (i.e., $\theta' \neq \theta$), if $\theta'=C^{\text{min}'}_L \cdot C_L \mod (1+S)$ and $\theta=C^{\text{min}}_L \cdot C_L \mod (1+S)$, respectively.
		
		Here, we use a contradiction proof by assuming if $C^{\text{min}}_L \neq C^{\text{min}'}_L$ and then $\theta' = \theta$ is true. This result implies that
		\begin{equation*}
			\begin{aligned}
				\theta'=\theta &\Rightarrow \theta' \equiv \theta \mod (1+S) \\
				&\Rightarrow C^{\text{min}'}_L (r) \equiv C^{\text{min}}_L(r) \mod (1+S) \\
				&\Rightarrow C^{\text{min}'}_L = C^{\text{min}}_L \cdot (\rightarrow \leftarrow).
				\label{col2.1}
			\end{aligned}
		\end{equation*}
		This contradicts to the assumption of $C^{\text{min}}_L \neq C^{\text{min}'}_L$. As a result, each $C^{\text{min}}_L \in [0,..,S]$ has its own and individual congruence $\theta \in [0,..,S]$.
		
		From this result, we can conclude that when $C^{\text{min}}_L=0,1,..,S$, it totally has $(1+S)$ possible outcomes of $n$ from Eq. \eqref{col1.5} due to its individual and difference congruence $\theta \in [0,..,S]$. According to Eq. \eqref{col1.5}, since difference and disjoint congruence $\theta$ will constitute different and disjoint sensor arrivals of $n$, i.e.,
		\begin{equation*}
			\begin{aligned}
				n=k(1+S)-\theta+1, k \in \mathbb{N}^+, \theta \in [0,...,S],
				\label{col2.2}
			\end{aligned}
		\end{equation*}
		that means that different series of $n$ for $(1+S)$ difference $\theta$ will be generated. Thus, all possible transmission arrivals $n\in \mathbb{N}^+$ will be constituted and any arrival can be received by a certain $C^{\text{min}}_L\in [0,...,S]$.
	\end{proof}
	\begin{lem}
		\label{lemma3}
		When the system has a finite waiting time, the worst delay will be incurred by such sensor $T_n=C_L-S$ (where $C_L-S>0$), i.e., if $C_L$ and $[W+S]$ are relative prime, and then $T^{max}=C_L-S$.
		
	\end{lem}
	\begin{proof}[\bfseries{Proof}]
		According to \emph{Corollary} \ref{cor1}, we know that for all $n=1,2,…,C_L$, their corresponding $C^{\text{min}}_L$ will not exceed $S$, i.e., $0\leq C^{\text{min}}_L\leq S$. In the following, we prove this lemma by contradiction. Assume it exists a slot $n'\neq C_L-S$, which insures a delay worse than $n=C_L-S$. We consider this in two cases: one is $n'<C_L-S$ and the other is $n'>C_L-S$ as follows.
		
		If $n'<C_L-S$, assume it has the largest value $C^{\text{min}'}_L=S$ to incur the maximal delay with the maximal $n'=C_L-S-1$ (because $n'<C_L-S$). From Eq. (1), we have
		\begin{equation*}
			\begin{aligned}
				n'+C^{\text{min}'}_L \cdot C_L&=(C_L-S-1)+SC_L\\
				&=C_L+SC_L-S-1\\
				&<C_L+SC_L-S\\
				&=(C_L-S)+SC_L\\
				&=n+C^{\text{min}}_L \cdot C_L,
				\label{1.3_1}
			\end{aligned}
		\end{equation*}
		which means that it is less than the delay incurred by $n=C_L-S$ with $C^{\text{min}}_L=S$ (by applying the results of \emph{Lemma} \ref{lemma2}). This is contradicted to the assumption.	
		
		For $n'>C_L-S$, we first combine it with $n=C_L-S$ and represent them as the continuous integers $n'=C_L-\delta$ (where $\delta=S,S-1,..,0$). According to Eq. \eqref{col1.2}, since $n'$ are $(1+S)$ continuous integers, it will have $(1+S)$ different remainders when $n'$ is divided by $(1+S)$ and thus results in $(1+S)$ different congruences $\theta \in [0,1,..,S]$. This means that each $n'$ has an individual $C^{\text{min}'}_L \in [0,1,..,S]$. Since $n=C_L-S$ has been proved in \emph{Lemma} \ref{lemma2} that it is with $C^{\text{min}}_L=S$, it means that the other $n'=C_L-\delta$ (where $\delta=S-1,..,0$, except for $S$) has the largest value $C^{\text{min}'}_L=S-1<S$.
		
		Therefore, if we want to find another worst delay from Eq. \eqref{tmax}, we have to use the maximum $n'=C_L$ with the maximum possible value $C^{\text{min}'}_L=S-1$. Then, the delay will be
		\begin{equation*}
			\begin{aligned}
				n'+C^{\text{min}'}_L \cdot C_L&=C_L+(S-1)C_L\\
				&=SC_L\\
				&<SC_L+C_L-S\\
				&=(C_L-S)+SC_L\\
				&=n+C^{\text{min}}_L \cdot C_L,
			\end{aligned}
		\end{equation*}
		which means that it is less than the delay incurred by $n=C_L-S$ with $C^{\text{min}}_L=S$ (by applying the results of \emph{Lemma} \ref{lemma2}). This is contradicted to the assumption. Thus, the worst delay will be incurred by $n=C_L-S$. 
	\end{proof}
	\begin{lem}
		\label{lemma4}
		The system has an "infinite" waiting time to receive all the sensors if $C_L$ and $[W+S]$ are not relative prime, i.e., if $GCD(C_{L}, [W+S]) \neq 1$, it exists at least one slot $(\exists n, 1\leq n \leq C_L)$ that cannot satisfy $n+C^{\text{min}}_L\cdot C_L= W+N^{\text{min}}_{sleep}[W+S]$.
	\end{lem}
	\begin{proof}[\bfseries{Proof}]
		Since $GCD(C_{L}, [W+S]) \neq 1$, it implies the existence of a greatest common factor $k>1(k \in \mathbb{Z}^+)$ that satisfies  $k|C_L$ and $k|[W+S]$. Then, we rewrite Eq. \eqref{tmax} as follows:
		
		\begin{equation}
			\begin{aligned}
				n-W= N^{\text{min}}_{sleep}[W+S]-C^{\text{min}}_L \cdot C_L.
				\label{p10}
			\end{aligned}
		\end{equation}
		In Eq. \eqref{p10}, the right hand side has a factor of $k$ due to $k|C_L$ and $k|[W+S]$. Thus, the left hand side of Eq. \eqref{p10} should have a factor of $k$ as well according to the prime factorization, i.e.,
		\begin{equation*}
			\begin{aligned}
				k \mid {n-W},
				\label{p11}
			\end{aligned}
		\end{equation*}
		and thus $\frac{n-W}{k} \in \mathbb{Z}$. However, since $n$ ranges from $1$ to $C_L$, these exists at least one value does not satisfy Eq. \eqref{tmax}. For example, $n=W+1$ cannot satisfy because $k \nmid 1$ for any $k>1$. Thus, it cannot ensure all possible transmissions to be received and incurs an “infinite” waiting time.
	\end{proof}
	
	Now, we investigate the worst transmission delay for such sensor $T^{max}$, denoted as $W_{delay}$, and then extend this analysis to evaluate the expected total worst transmission delay, denoted as $E[W_{delay}^{total}]$, considering the conditional probability $Pr(c)$.
	\begin{lem}
		\label{lemma5}
		The worst transmission delay $W_{delay}$ incurring by the wakeup-sleep operation can be obtained by $W_{delay}=C_L(S+1)-S$.
	\end{lem}
	\begin{proof}[\bfseries{Proof}]
		According to \emph{Lemma} \ref{lemma2}, we have the worst sensor transmission $T^{max}=n'=C_L-S$ with $C^{\text{min}}_L=S$. Since $W=1$, Eq. \eqref{tmax} can be written as follows:
		\begin{equation*}
			\begin{aligned}
				W_{delay}&=n'+C^{\text{min}}_L \cdot C_L\\
				&=(C_L-S)+S \cdot C_L\\
				&=S \cdot C_L+C_L-S\\
				&=C_L(S+1)-S.
			\end{aligned}
		\end{equation*}
	\end{proof}
	
	\begin{theorem} For any sensor transmission $T_n$, the system has a total expected worst transmission delay $E[W_{delay}^{total}]$ with $t$ times collision, where
		\begin{equation}\begin{aligned} 
				E[W_{delay}^{total}]= (C_L(S+1)-S) {\left(\frac{C_L-1}{C_L}\right)}^{(N-1)} \sum_{t=0}^{\infty} (t+1) \\ \left(1-{\left(\frac{C_L-1}{C_L}\right)}^{(N-1)}\right)^{t}.
				\label{thr02}
			\end{aligned}
		\end{equation}
		
		\label{th002}
	\end{theorem}
	\begin{proof}[\bfseries{Proof}]
		Suppose we have $N$ sensors to transmit in the system and each sensor randomly chooses a time slot to transmit within a duty-cycle in the first round without pre-defined rules or synchronization, the probability of collision $Pr(c)$ for the worst sensor transmission $T^{max}$ being received at $W_{delay}$ slot is given by
		\begin{equation}
			\label{eq16}
			Pr(c)=1-Pr(s),
		\end{equation}
		where $Pr(s)$ is the successful probability without collision. To calculate $Pr(s)$, we consider the case that any other sensors choose a slot different than that one chosen by sensor $T^{max}$. Since there are $N$ sensors and $C_L$ time slots available for all transmissions, the probability that sensor $T^{max}$ arrives but the remaining $N-1$ sensors choose other time slots will be
		\begin{equation*}
			\begin{aligned}
				Pr(s)&={\left(\frac{C_L-1}{C_L}\right)}^{(N-1)} \\& =1-Pr(c).
			\end{aligned}
		\end{equation*}
		Thus, if sensor $T^{max}$ is not received at time slot $W_{delay}$ and poses $t$ times of collisions before being received successfully, it will incur a total delay consisting of the sum of $t$ previous collision delays and the last successful reception delay. Mathematically, this total delay is given by
		\begin{equation*}
			\begin{aligned}
				W_{delay}^t=t\times W_{delay}+W_{delay}=(t+1)\times W_{delay}.
			\end{aligned}
		\end{equation*}
		The probability of this total delay\footnote{
			Note that if the system does not receive a sensor over $T_{delay}$ slots, the system will re-activate such sensor to randomly transmit again to avoid such sensor being collided by the same pattern.}, denoted as $Pr(W_{delay}^t)$, is determined by the probability of encountering a collision $Pr(c)$, raised to the power of $t$, multiplied by the probability of successful reception $Pr(s)$, and it is with a probability given by
		\begin{equation*}
			\begin{aligned}
				Pr(W_{delay}^{t})=Pr(c)^t \times Pr(s).
			\end{aligned}
		\end{equation*}
		Therefore, using this result and \emph{Lemma} \ref{lemma5}, we can derive the total expected transmission delay, denoted as $E[W_{delay}^{total}]$, as follows:
		\begin{equation*}\begin{aligned}
				E[W_{delay}^{total}]&= \sum_{t=0}^{\infty} W_{delay}^{t} \cdot Pr(W_{delay}^{t})\\
				&=\sum_{t=0}^{\infty}(t+1)\times W_{delay} \cdot Pr(c)^t \times Pr(s) \\
				&=W_{delay}Pr(s)\sum_{t=0}^{\infty}(t+1) \cdot ({(1-Pr(s))}^t)\\
				&=(C_L(S+1)-S) {\left(\frac{C_L-1}{C_L}\right)}^{(N-1)} \\ & \quad \quad \quad \sum_{t=0}^{\infty} (t+1) \left(1-{\left(\frac{C_L-1}{C_L}\right)}^{(N-1)}\right)^{t}.
			\end{aligned}
		\end{equation*}
	\end{proof}
	In conclusion, \emph{\bf{Theorem \ref{th001}}} and \emph{\bf{Lemmas \ref{lemma1} $\sim$ \ref{lemma4}}} are used to determine that the system has a finite delay or not, while \emph{\bf{Theorem \ref{th002}}} and \emph{\bf{Lemmas \ref{lemma5}}} are used to calculate the total expected worst transmission delay $E[W_{delay}^{total}]$.
	\begin{figure} [t]
		\centering
		\includegraphics[width=\linewidth]{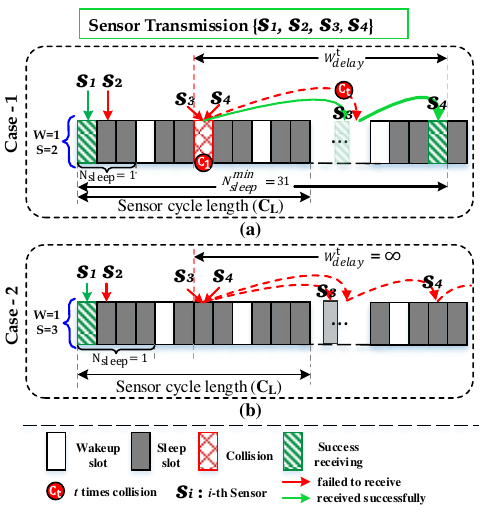}
		\caption{Example of two transmission scenarios}
		\label{Scheme}
	\end{figure}
	
	Below, we give two examples in Fig.~\ref{Scheme} to clarify the analytic results. In the Fig.~\ref{Scheme}(a), considering $GCD(C_{L}, [W+S]) = 1$, we assume that $C_{L}=32$ and $[W+S]=[1+2] = 3$. Based on the \emph{analytic scheme}, because $GCD(C_L, [W+S]) = 1$, we have $T^{\text{max}} = C_L - S = 32 - 2 = 30$. This implies that the sensor arriving at the $30^{\text{th}}$ slot experiences the longest transmission delay of $W_{delay}=C_L(S+1)-S=32\times(1+2)-2=94$ with $C^{\text{min}}_L=S=2$ and $N^{\text{min}}_{sleep} =C_L-1=31$.

	On the other hand, in Fig.~\ref{Scheme}(b), considering 
	$[W+S] = [1+3]=4$. Since $GCD(C_{L}, [W+S]) \neq 1$,
	the worst transmission $T^{\text{max}}$ does not exist. Therefore, it has an infinite waiting time to receive all sensors in such system.
	\section{Performance Evaluation} 
	In this section, we conduct both the field trial experiments and mathematical analysis to validate the effectiveness of the system. The performance evaluations, collectively referred to as \textbf{Mathematical Analysis} and \textbf{Real-time Experiment} with different parameters of $S=2,6,30$ where $GCD(C_{L},[W+S])=1$ and $W$ is fixed as $1$, to demonstrate the efficiency and effectiveness of our proposed system. It is important to note that there has been no previous work specifically addressing TP sensor with sleep operation and transmission delay, making our study the first to comprehensively investigate this critical aspect of TP sensor performance.
	
	To simulate a realistic range of TP sensor counts for passenger vehicles and heavy-weight logistic trucks, trailers, and flatbeds, we considered $N=4 \sim 32$ sensors in our simulations. This range encompasses the typical number of sensors found on passenger vehicles ($4$ to $8$ wheels) \cite{wheel} and heavy-weight logistic trucks, trailers, and flatbeds ($8$ to $16$ wheels) \cite{wheel}. Note that the experiments will demonstrate the performance results even if the system is overloaded (i.e., $N \approx C_L$, where $C_L=32$ \cite{r_sensor}).
	
	\subsection{Worst Transmission Delay}
	Fig.~\ref{wd} shows the worst transmission delay (in seconds) experienced by the system for different numbers of sensors transmitting. As the number of transmitting sensors increases, the delay also increases due to frequent collisions of sensor transmissions. Note that all the pair results of mathematical and experiment analysis are consistent.  
	Specifically, \textbf{Mathematical Analysis [S=2]} shows the lowest transmission delay due to a higher number of total wakeup slots per cycle length, allowing more data to be received in the system.
	\begin{figure} [!h]
		\centering
		\includegraphics[width=\linewidth]{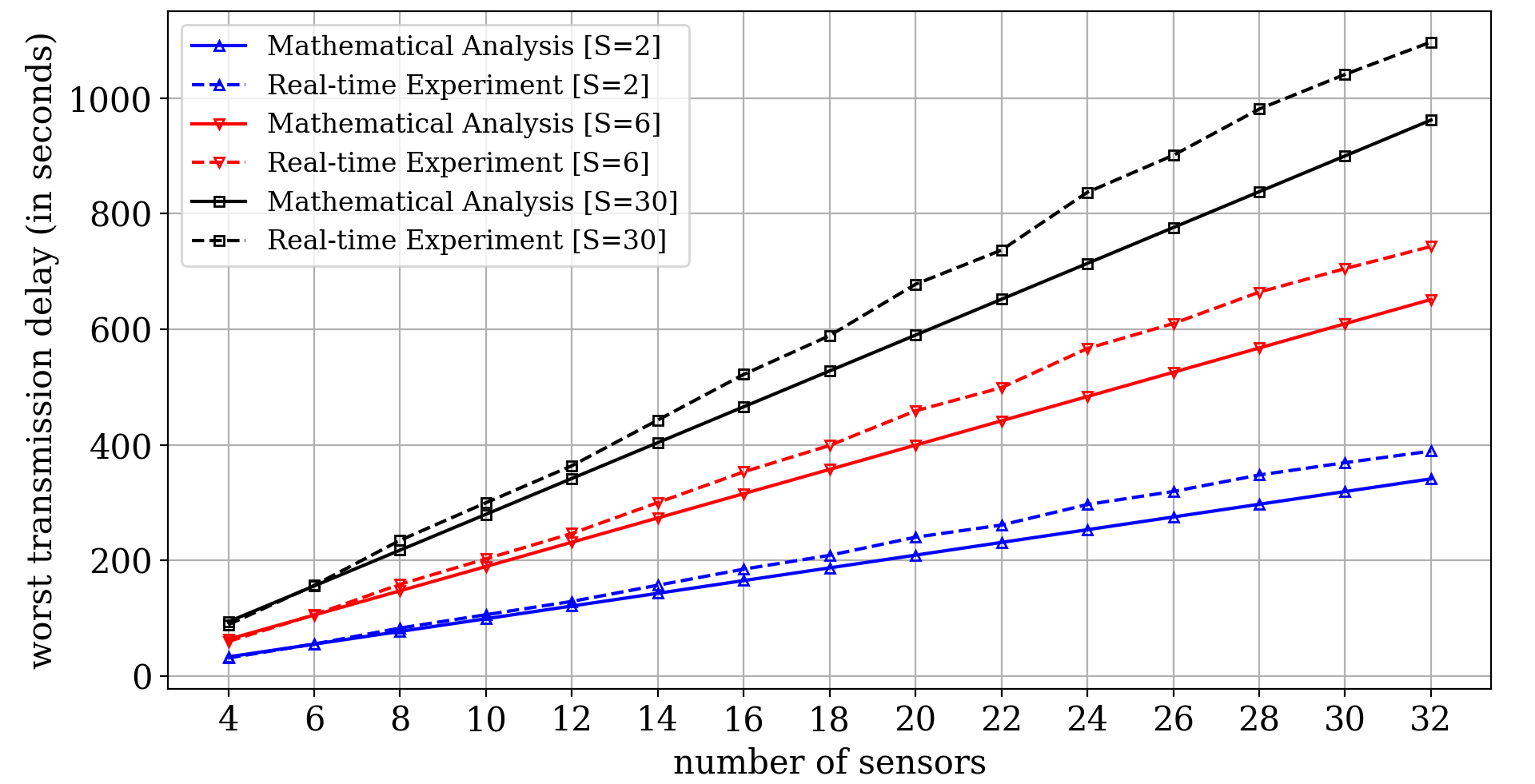}
		\caption{Comparisons on worst transmission delay}
		\label{wd}
	\end{figure}
	\subsection{Expected Worst Transmission Delay} Fig.~\ref{ed} shows the expected worst transmission delay (in seconds) received in the system for different numbers of sensors transmitting, where $E[W_{delay}^{total}]=\sum_{t=0}^{\infty}(t+1)\times W_{delay} \cdot Pr(c)^t \times Pr(s)$. Based on different numbers of transmitting sensors in the system, it is shown that the expected delay of \textbf{Real-time Experiment} is always higher than \textbf{Mathematical Analysis}. This is due to a higher collision probability $(Pr(c))$ with an increasing number of transmitting sensors, resulting in a higher expected worst transmission delay. 
	\begin{figure} [!h]
		\centering
		\includegraphics[width=\linewidth]{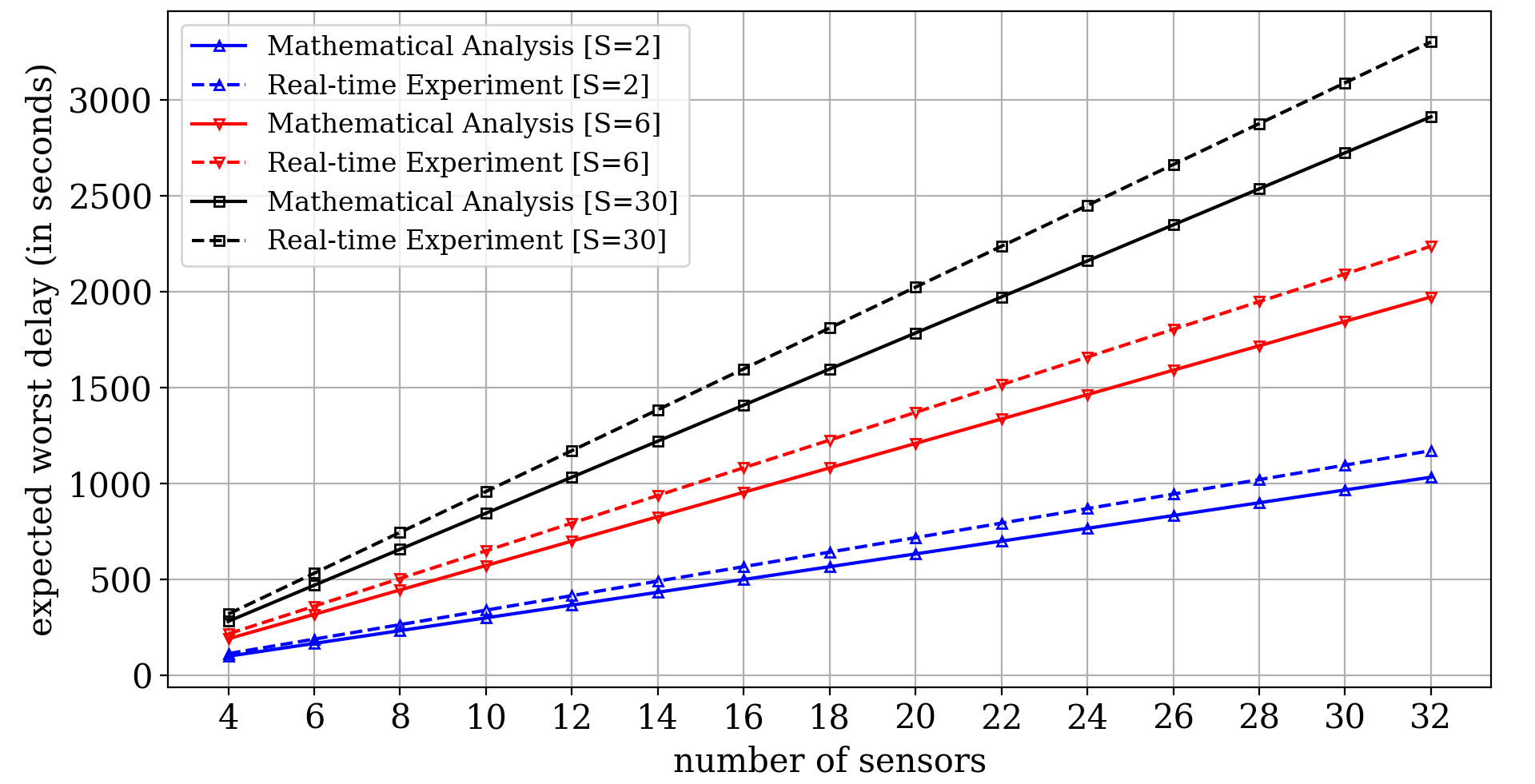}
		\caption{Comparisons on expected worst transmission delay}
		\label{ed}
	\end{figure}
	\subsection{Average Transmission Delay}Fig.~\ref{ad} shows the average transmission delay (in seconds) received in the system for different numbers of sensors transmitting. \textbf{Real-time Experiment} results show the worst average results as they're based on real-time data with higher collision probabilities, which causes a higher average transmission delay when considering the higher number of sensors transmitting per $C_{L}$. \textbf{Mathematical Analysis [S=2]} results are the lowest and most efficient due to the maximum number of total wakeup slots per cycle length among all other compared schemes.
	
	From Fig. \ref{wd} to Fig. \ref{ad}, we can observe that for vehicles with fewer tires (e.g., $4$ to $8$ wheels), a shorter sleep cycle length (e.g., $S = 2$) can facilitate the reception of all sensor data. For those special large vehicles such as trucks, trailers, and flatbeds with a larger number of wheels (e.g., $\ge 12$ wheels), the individual receivers will be recommended for the detached ones to facilitate the data reception.
	\begin{figure} [!h]
		\centering
		\includegraphics[width=\linewidth]{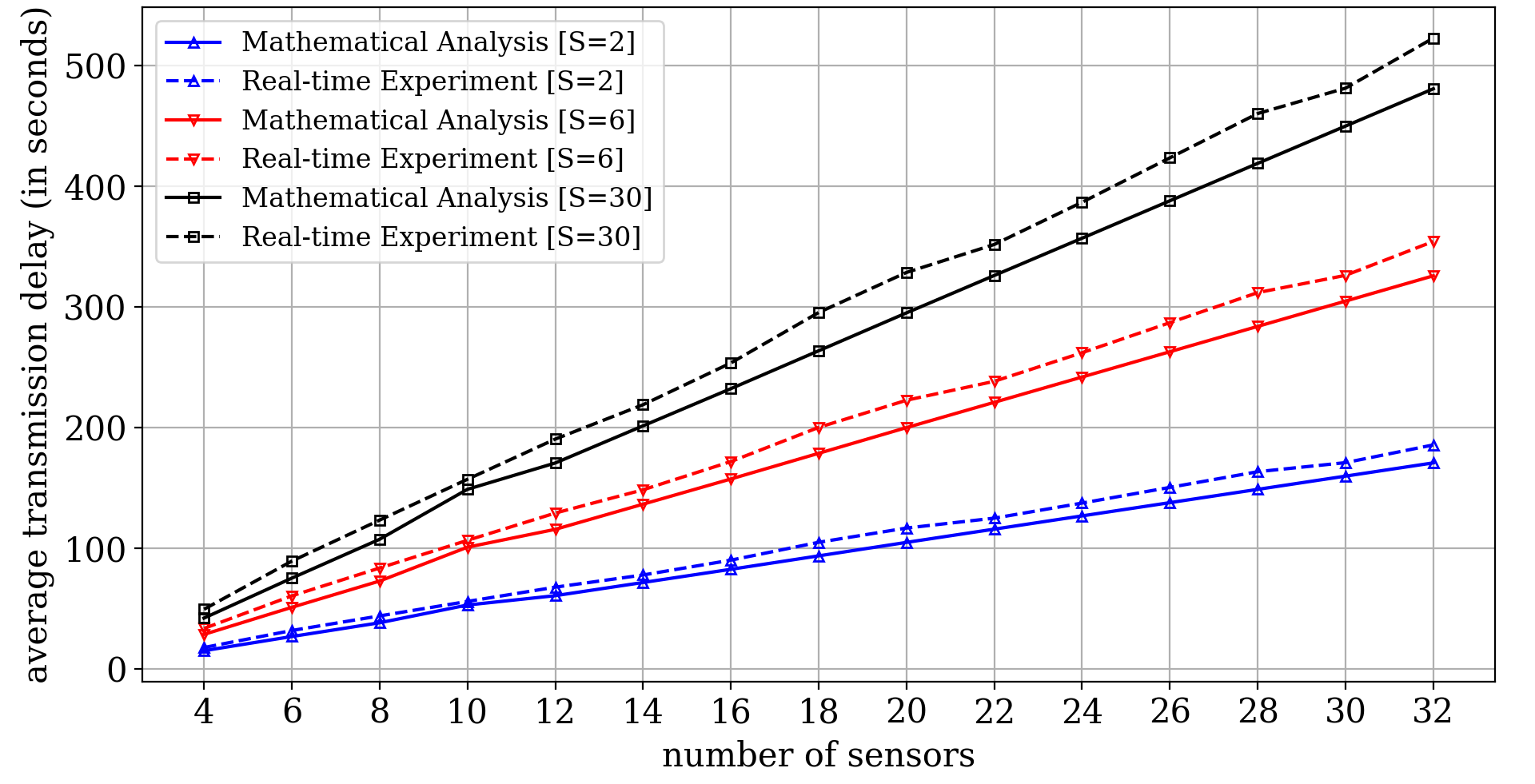}
		\caption{Comparisons on average worst transmission delay}
		\label{ad}
	\end{figure}
	\subsection{Power Saving Ratio}The power saving ratio is calculated based on various numbers of transmitting sensors, where ${\textit{power saving ratio (\%)}= \frac { \textit{total sleep period (S)}}{\textit{total wakeup-sleep period (W+S)}}\times 100}$. As depicted in Fig.~\ref{pw}, all the schemes exhibit power-saving ratios ranging from 66.7\% to 96.7\% for S=2 to S=30, respectively. Even as the number of sensors increases, the power-saving ratio remains constant for each specific value of S due to the static nature of wake-up slots in the system. 
	\textbf{Real-time Experiment [S=2]} demonstrates the lowest power saving ratio among all the schemes, primarily due to the fewest number of sleep slots.
	
	\begin{figure} [!h]
		\centering
		\includegraphics[width=\linewidth]{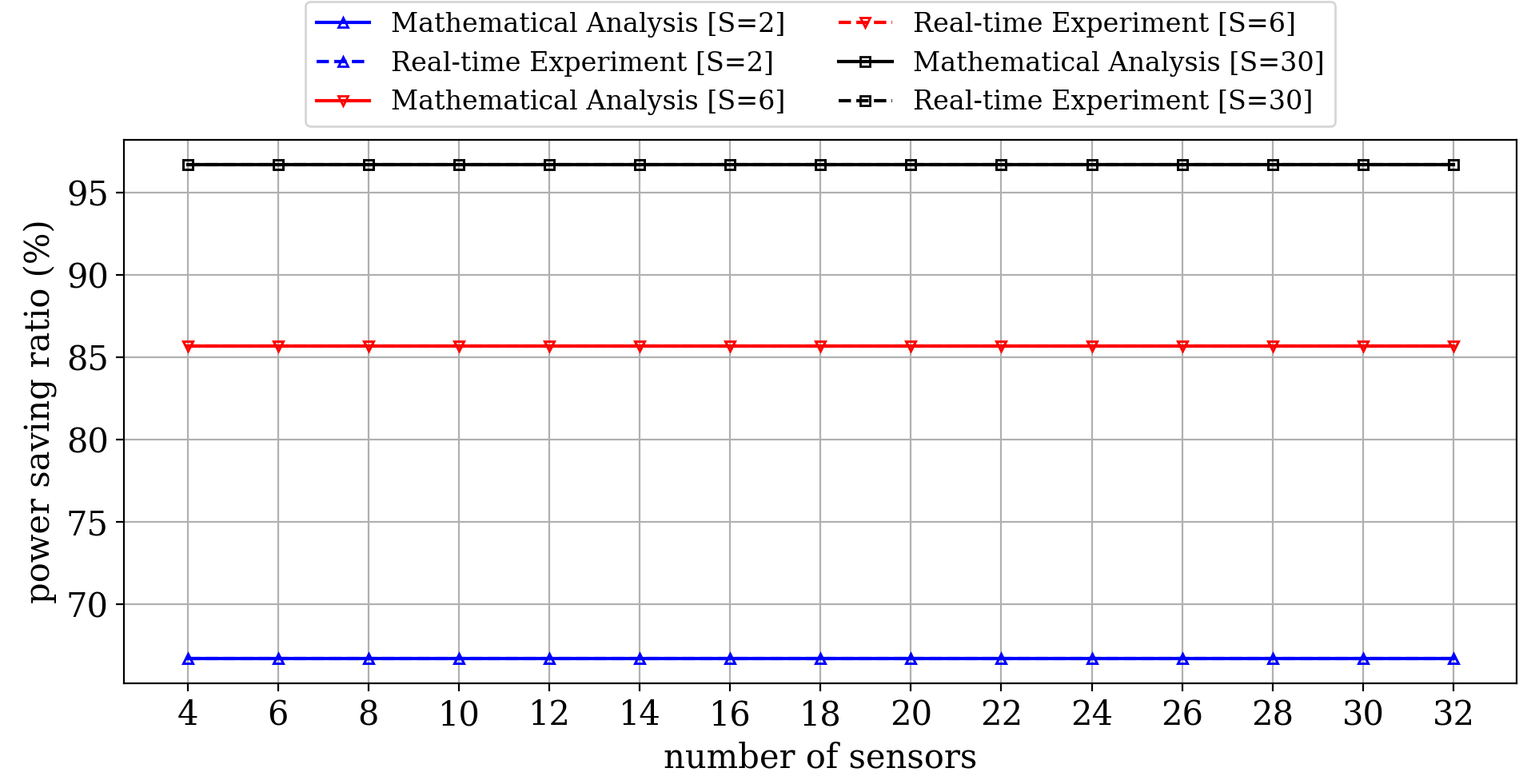}
		\caption{Comparisons on power saving ratio}
		\label{pw}
	\end{figure}
	
	\begin{table*}[!htbp]
		\caption{Comparison and Analysis with Other Applications}
		\label{applications}
		\begin{tabularx}{\textwidth}{@{} l *{9}{C} c @{}}
			\toprule
			\bf{Applications} & \bf{Sensors} & \bf{Wireless Channel} & \textbf{Trans. Interval \quad\quad\quad(in ms)}  
			& \bf{Delay Requirements (in ms)} & \multicolumn{2}{c}{\bf{Energy Consumption}} & \multicolumn{2}{c}{\bf{*Suggested Parameters}}   \\
			\cmidrule(lr){6-7}
			\cmidrule(lr){8-9}
			& & & & & \textbf{Energy Issue?} & \textbf{Applicable or not?} & \textbf{$C_L$} & \textbf{$S$} \\
			\midrule
			ABS \cite{FN3} & Wheel inertial & 5 GHz & 100$\sim$200 
			& 150 & Yes & Yes & 150 & 10 \\ \hline
			
			ESC \cite{t5} & Wheel inertial & 5 GHz & 100$\sim$300 
			& 200 & Yes & Yes & 200 & 12 \\ \hline
			
			PAS \cite{t4} & Ultrasonic & 433 MHz & 200$\sim$500 
			& 350 & Yes & Yes & 350 & 22   \\ \hline
			
			OSS \cite{MC4_1} & PIR & 5 GHz & 100$\sim$300 
			& 200 & Yes & Yes & 200 & 12 \\ \hline
			
			LDW \cite{R3} & Radar, Infrared & 2.4 GHz & 100$\sim$500 
			& 300 & Yes & Yes & 300 & 20  \\ 
			
			\bottomrule
		\end{tabularx}
		\begin{tablenotes}
			\item[] \textbf{* }The parameters are based on proposed analytic scheme.
			\item[]
			\textbf{ABS:} Anti-lock Braking System;
			\textbf{ESC:} Electronic Stability Control; 
			\textbf{PAS:} Parking Assist Systems; 
			\textbf{OSS:} Occupancy Sensor System; 
			\textbf{LDW:} Lane Departure Warning.
		\end{tablenotes}
	\end{table*}
	
	\subsection{Cumulative Correctness Rate (CCR)} Finally, the correctness rate, the ratio of the theoretical values derived from the \textbf{Mathematical Analysis} (Fig.~\ref{cr}) to the corresponding \textbf{Real-time Experiment} values for \emph{Worst Transmission Delay} and \emph{Average Transmission Delay}. CCR is defined as:
	\begin{equation*}
		\begin{aligned}
			\textit{CCR}=\frac{\overline{\textit{MA}}\times 100} {\overline{\textit{RE}}}
		\end{aligned}
	\end{equation*}
	where $\overline{\textit{MA}}$ and $\overline{\textit{RE}}$ represent the average values of the three aforementioned metrics in the \textbf{Mathematical Analysis} and the \textbf{Real-time Experiment}, respectively. 
	The \textbf{Mathematical Analysis} achieved a significantly higher CCR (over $97.9\%$), demonstrating the analytic scheme's consistency in replicating real-world behavior and its overall reliability in delivering accurate values.
	\begin{figure} [!h]
		\centering
		\includegraphics[width=\linewidth]{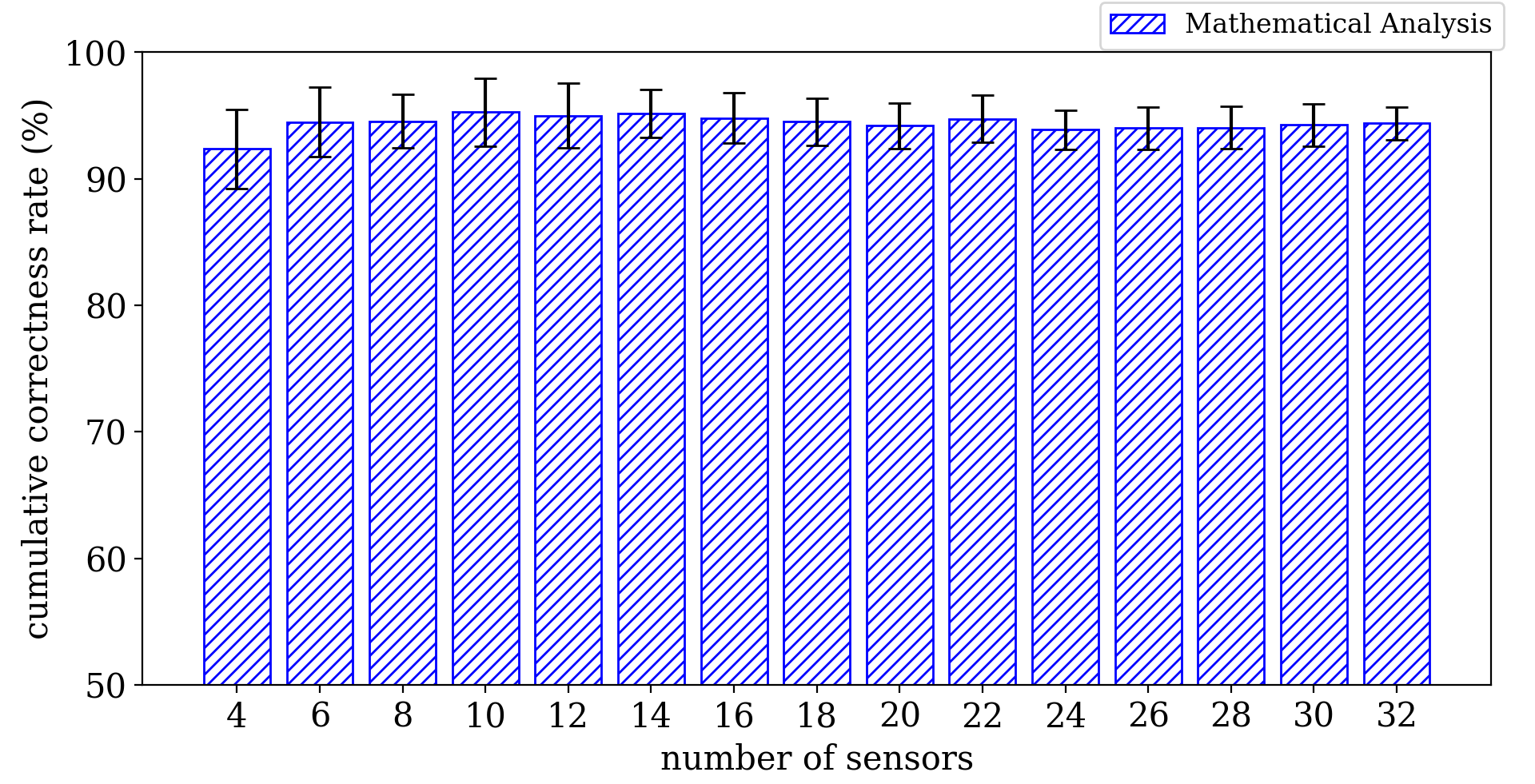}
		\caption{Correctness rate of the proposed scheme}
		\label{cr}
	\end{figure}
	
	\subsection{Discussion of Comparison and Analysis with Other Applications}
	Table \ref{applications} illustrates the potential applicability of the proposed system in various wireless sensor-based applications within intelligent vehicles, including the Anti-lock Braking System (ABS) \cite{FN3}, Electronic Stability Control (ESC) \cite{t5}, Parking Assist Systems (PAS) \cite{t4}, Occupancy Sensor System (OSS) \cite{MC4_1}, and Lane Departure Warning (LDW) \cite{R3}. We highlight a common challenge of high energy consumption resulting from frequent data transmissions (ranging from $100$ to $500$ slots, with each slot lasting milliseconds) across all applications. A key advantage of our proposed system lies in its adaptability to address such challenges.
	If adopting our approach on these applications, we suggest setting the parameters of the proposed scheme for ABS, ESC, PAS, OSS, and LDW at $150$, $200$, $350$, $200$, and $300$ slots, respectively,  with sleep lengths set as $10$, $12$, $22$, $12$, and $20$ slots, respectively, to potentially balance transmission delays and energy consumption, both critical factors safety-critical systems. By employing our analytic scheme to analyze data and determine delays, it has the potential to mitigate the issue of high energy consumption while ensuring the seamless operation of applications.
	
	\section{System Demonstration}
	Intelligent vehicles rely on the IoT-enabled WTSS, a critical component that utilizes various IoT and wireless technologies to facilitate effective and efficient real-time monitoring. This section delves into its applications and demonstrates its deployment across diverse IoT and wireless technologies to achieve real-time monitoring capabilities. Note that Table \ref{tab1} summarizes a comparative analysis of existing wireless tire sensing systems. We can say that our system has more features and more comprehensive functions than traditional systems.
	\begin{table*}
		\caption{Comparison of existing systems}
		\label{tab1}
		\begin{tabularx}{\textwidth}{@{} l *{10}{C} c @{}}
			\toprule
			\bf{Ref.\& Year} & \bf{Wireless technology} & \bf{Network type} & \bf{Power consumption} & \bf{Wireless range} & \bf{Live geo-location} & \bf{Storage  (local/cloud)} & \bf{Energy conservation} & \bf{Cost} & \bf{Delay analysis} \\ 
			\midrule
			2019 \cite{t1} & Bluetooth & WPAN & Low & Short & No & Local & No & Low & No  \\ \hline
			2019 \cite{t2} & Wi-Fi & WLAN & High & Medium & No & Cloud & No & High & No  \\ \hline
			2020 \cite{t3} & LoRaWAN & LoRaWAN & High & Short & Yes & Cloud & No & High & No  \\ \hline
			2021 \cite{t4} & Cellular & WWAN & High & Long & Yes & Cloud & No & Medium & No  \\ \hline
			2021 \cite{c11} & - & - & High & - & No & - & No & Medium & No  \\ \hline
			2022 \cite{c18} & Bluetooth & WPAN & Low & Short & No & Local & No & Medium & No  \\ \hline
			2022 \cite{t5} & Zigbee & WPAN & Low & Short & No & Cloud & No & Medium & No  \\ \hline
			\textbf{\tc{Proposed System}} & \textbf{\tc{Bluetooth, Wi-Fi \& Cellular}} & \textbf{\tc{WPAN, WLAN \& WWAN}} & \textbf{\tc{Low}} & \textbf{\tc{Long}} & \textbf{\tc{Yes}} & \textbf{\tc{Local \& Cloud}} & \textbf{\tc{Yes}} & \textbf{\tc{Medium}} & \textbf{\tc{Yes}}  \\
			\bottomrule
		\end{tabularx}
	\end{table*}
	
	\begin{figure} [t!]
		\centering
		\includegraphics[width=\linewidth]{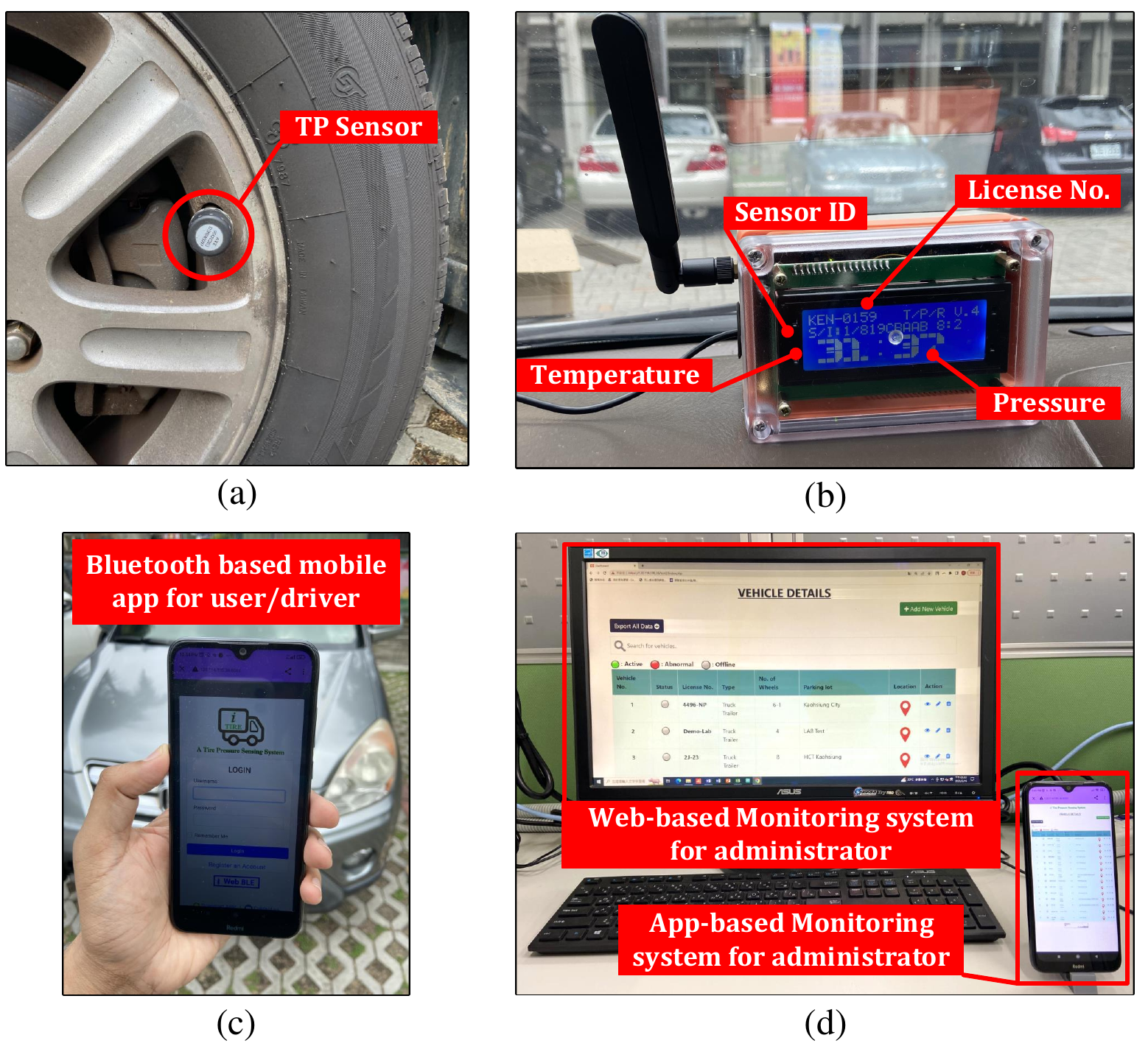}
		\caption{Working scenario of proposed IoT-enabled WTSS: (a) wireless TP sensor installed on the tire of a vehicle, (b) TP Gateway installed in the vehicle, (c) bluetooth based mobile app for user/driver and (d) web and mobile app based monitoring system for administrators}
		\label{d_1}
	\end{figure}
	
	\subsection{Overview and Working Scenario of IoT-Enabled WTSS} The proposed IoT-enabled WTSS system can manage real-time data for intelligent vehicles, making it suitable for organizations with small or large parking spaces. This includes residential neighborhoods, local parks, universities, airports, stadiums, business parks, and large automobile industries.
	Fig.~\ref{d_1} illustrates the IoT-enabled WTSS system, which utilizes multiple wireless technologies for various purposes, including 1) Bluetooth for mobile app connectivity, 2) Wi-Fi for industrial parking areas, and 3) 4G cellular for long-distance communication and real-time vehicle location tracking. The system captures and processes tire pressure and temperature data from wireless TP sensors installed on each tire of the intelligent vehicle, which is transmitted wirelessly to the TP Gateway for decoding and storage in the proposed cloud-based management system's database. The system's GUI dashboard provides real-time and historical data analysis, including vehicle type, parking lot, sensor positions, live geo-location, and other insights.
	\begin{figure} [!h]
		\centering
		\includegraphics[width=\linewidth]{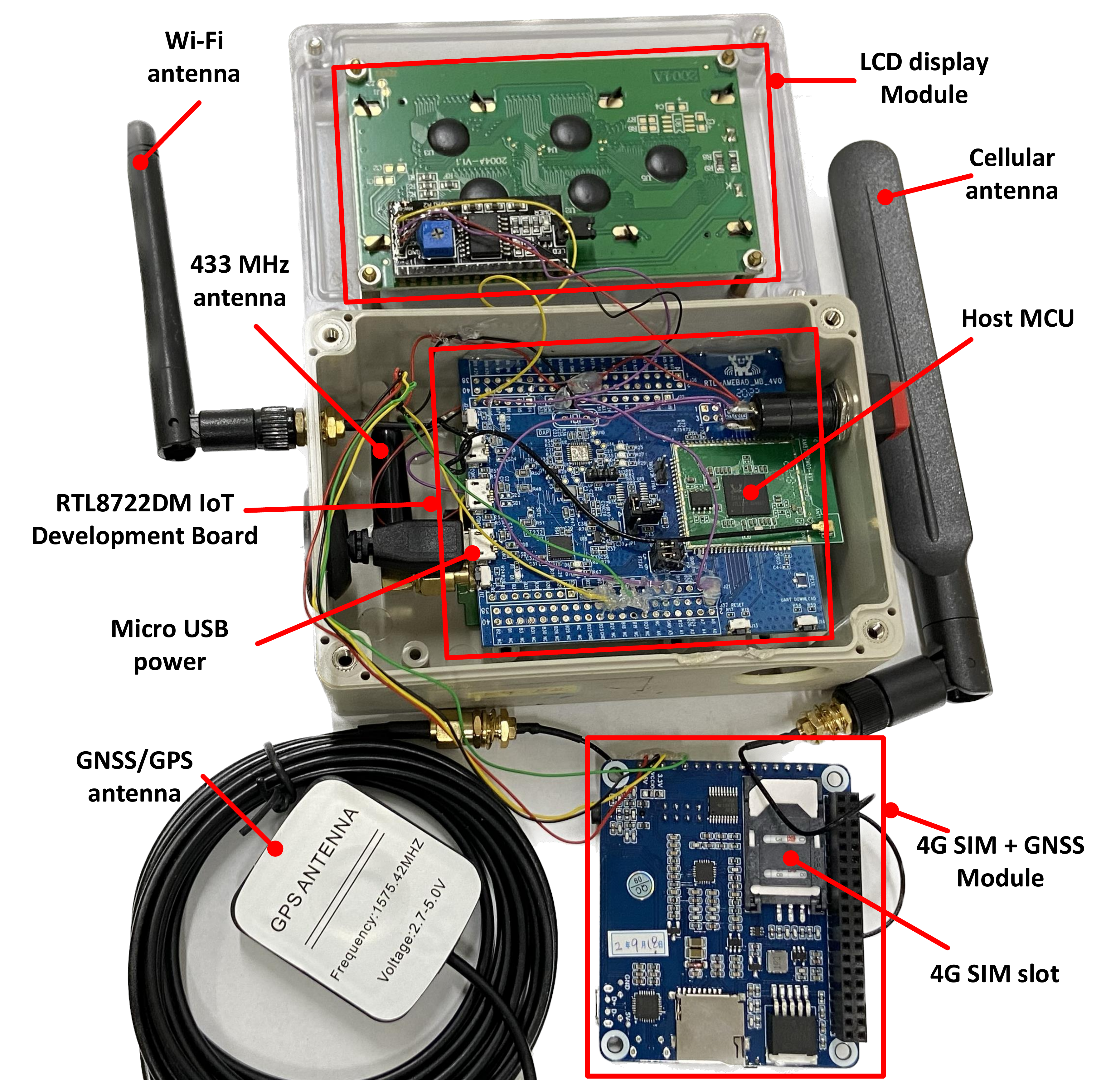}
		\caption{Various parts of TP Gateway}
		\label{d_2}
	\end{figure}
	
	\subsection{IoT-Enabled Modules and Components}The prototype of our system has been developed and verified. Fig.~\ref{d_2} shows the various parts of the proposed TP Gateway. The TP Gateway is the data processing source of our proposed IoT-enabled WTSS system. To gather tire pressure, temperature, and current status of sensors, we have integrated the device with multiple modules, including RTL8722DM development board, I2C LCD, 4G SIM module, host microcontroller, 433 MHz antenna, Wi-Fi antenna, cellular antenna, and GNSS/GPS antenna.
	\begin{figure} [!h]
		\centering
		\includegraphics[width=\linewidth]{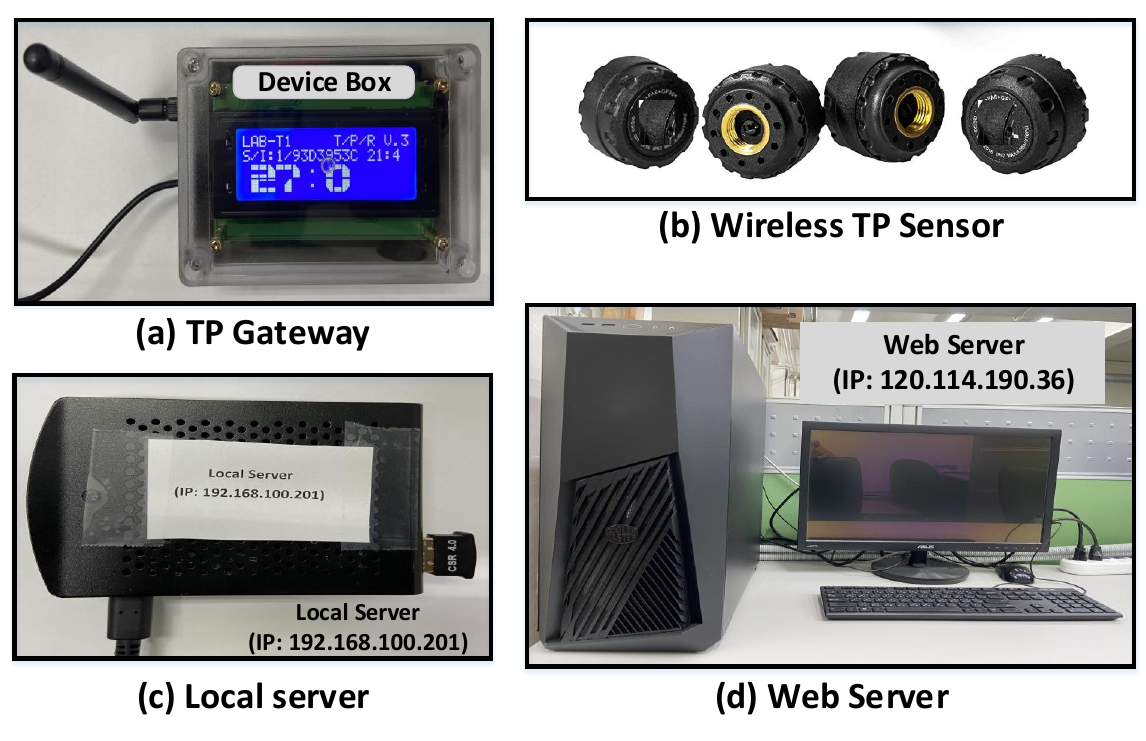}
		\caption{The hardware components: (a) TP Gateway, (b) wireless TP sensor, (c) raspberrypi as a local server, and (d) web server}
		\label{d_3}
	\end{figure}
	
	The portable device, shown in Fig.~\ref{d_3}(a), has a box shape and features various wireless technologies that allow it to communicate with multiple devices simultaneously. It also includes a push button for activating the I2C LCD display. Fig.~\ref{d_3}(b) illustrates 3V lithium battery-powered wireless TP sensors that gather tire pressure and temperature data with high accuracy. The module collects and transmits data through the radio frequency transmitter to the 433 MHz antenna. A local server based on Raspberry Pi4 and a web server based on Apache HTTP v3.3.0 is installed on our surveillance server. Fig.~\ref{d_3}(c) and (d) depict the hardware configuration of the local and web servers, respectively.
	
	The TP Gateway is an IoT-enabled real-time system designed to collect and process tire pressure and temperature data from wireless TP sensors installed on the tires of intelligent vehicles. The gateway is an important technology for the automotive industry as it enables real-time data processing, which leads to improved safety, better fuel efficiency, and longer tire life. By employing multiple wireless technologies for data transmission, it can be highly efficient and reliable. Overall, the TP Gateway can help prevent accidents and reduce maintenance costs, making it a valuable asset for fleet management and individual vehicle owners alike.
	\begin{figure} [!h]
		\centering
		\includegraphics[width=\linewidth]{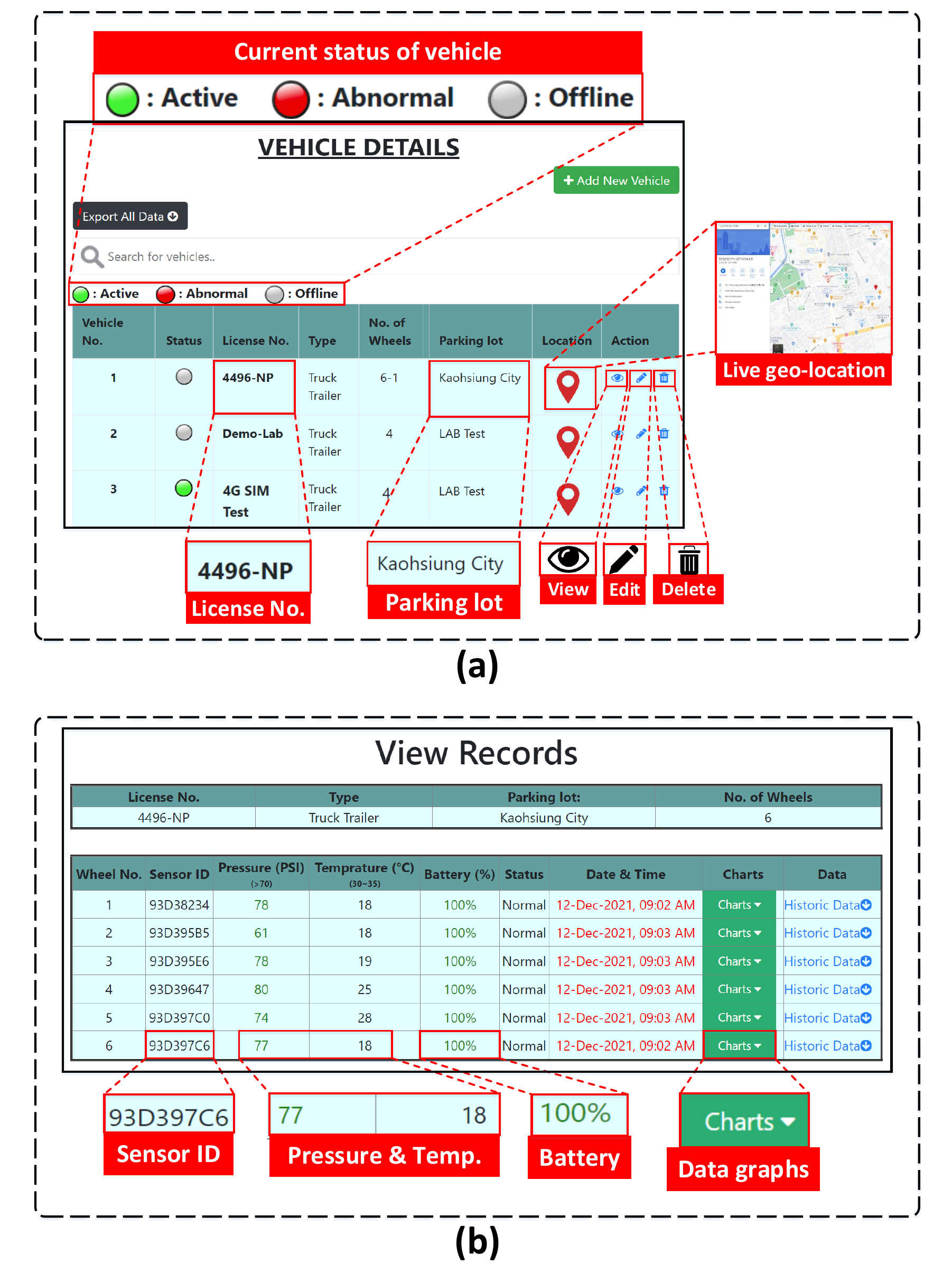}
		\caption{Snapshots of graphical user interface: (a) vehicle details including status and live location (b) vehicle's record including real-time pressure and temperature}
		\label{d_4}
	\end{figure}
	
	\begin{figure} [!h]
		\centering
		\includegraphics[width=\linewidth]{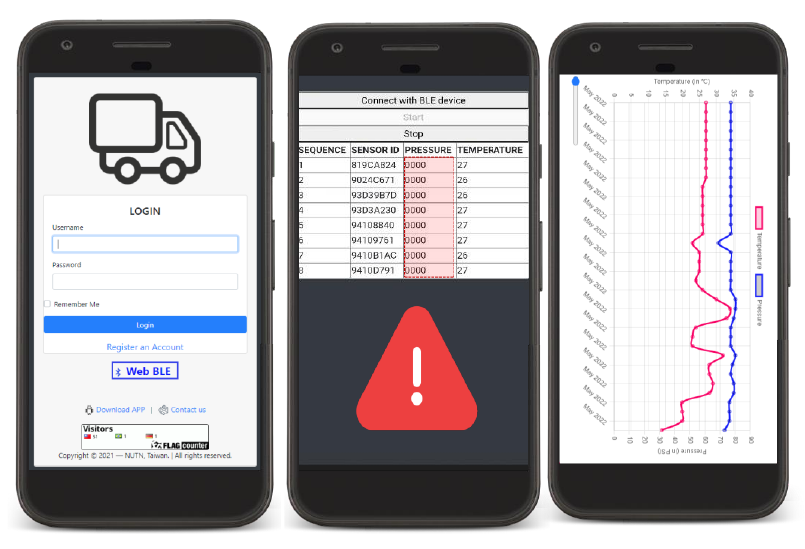}
		\caption{Snapshots of andriod based application}
		\label{d_5}
	\end{figure}
	
	\subsection{Cloud-Based Management System}Our graphical user interface (GUI) is developed in HTML, php, JavaScript and CSS. Fig.~\ref{d_4} shows different snapshots of GUI. Fig.~\ref{d_4}(a) shows all the listed vehicles with their status, vehicle type, number of wheels with attached sensors, parking lot, and live geo-location. The user interface allows the user to define a sensor by inserting its ID information and fetches real-time sensor information from the database while keeping it updated. Fig.~\ref{d_4}(b) displays the vehicle record for each individual vehicle, which includes sensor ID, pressure, temperature, battery status, accurate time of data update, and charts for historical data.
	
	Fig.~\ref{d_5} shows snapshots of an Android-based mobile app. The mobile app can fetch real-time data using Bluetooth technology without the internet or over the internet. The mobile app features an SOS emergency sound alarm for abnormal tire status detection to notify the user.
	
	\section{Conclusion and Future Work}
	In this paper, we have developed and validated an efficient IoT-enabled wireless tire sensing system (WTSS), addressing the issue of worst transmission delay for any sensor received while considering energy savings. Our analytic scheme effectively evaluates factors affecting transmission and identifies the worst transmission delay based on different wakeup-sleep periods and collision probabilities, leading to higher energy efficiency. The real-time WTSS system, with its wireless technologies, cloud-based management system, and mobile app, offers comprehensive data insights and instant alerts to avoid accidents and boost vehicle efficiency. The effectiveness of the proposed scheme is verified by analysis, experimentation, and simulation results. The scheme effectively addresses the worst transmission delay for sensor transmission, providing a reliable and efficient solution for real-time monitoring of tire status in intelligent vehicles. This will be valuable for developers working on real-time applications for logistic trucks where minimizing sensor transmission delays is crucial.
	
	For future work, we will further consider the energy issue of the transmitter in the WTSS. We plan to develop a new prototype of sensor transmitters with an energy-saving mechanism, enabling a more in-depth analysis of the relationship between energy consumption and delay. This initiative aims to extend the sensors’ lifetime while ensuring the immediacy of tire pressure monitoring.
	\bibliographystyle{IEEEtran}
	\bibliography{bibtii}

	\begin{IEEEbiography}[{\includegraphics[width=1in,height=1.25in,clip,keepaspectratio]{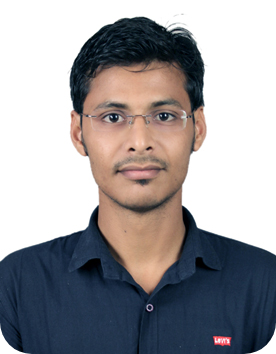}}]{Shashank Mishra} received the Masters degree in Software Technology from the VIT University, Vellore, India in 2013. He is currently working toward the Ph.D. degree in Electrical Engineering with the Department of Electrical Engineering, National University of Tainan, Tainan, Taiwan.
		
		His research interests include resource utilization and management in Artificial Intelligence of Things (AIoT), Wireless Sensor Networks (WSNs), and Intelligent Vehicles.
	\end{IEEEbiography}
	
	\begin{IEEEbiography}[{\includegraphics[width=1in,height=1.25in,clip,keepaspectratio]{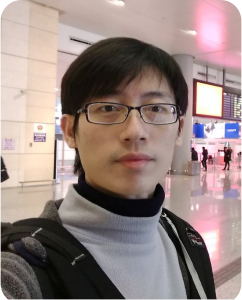}}]{Jia-Ming Liang} received his Ph.D. degree in Computer Science from National Chiao Tung University in 2011. He was an Assistant Research Fellow in the Department of Computer Science, National Chiao Tung University, Taiwan, from 2012 to 2015. In 2015, he was an Assistant Professor in the Department of Computer Science and Information Engineering, Chang Gung University, Taoyuan, Taiwan. Since 2020, he joined the Department of Electrical Engineering, National University of Tainan, Taiwan, where he is currently an Associate Professor.
		
		His research interests include resource management in 5G/B5G mobile networks and  Artificial Intelligence of Things (AIoT).
	\end{IEEEbiography}
\end{document}